\documentclass[10pt]{article}

\usepackage{amstext,amssymb,natbib,amsthm,booktabs,amsmath,graphicx,rotating,array}
\usepackage{algorithmic,algorithm}
\usepackage[margin=1in]{geometry}
\usepackage{natbib,url}

\newtheorem{theorem}{Theorem}[section]
\newtheorem{lemma}[theorem]{Lemma}

\newtheorem{remark}[theorem]{Remark}

\newtheorem{proposition}[theorem]{Proposition}
\newcommand{\ba}{\boldsymbol{a}}

\newcommand{\be}{\boldsymbol{e}}

\newcommand{\bu}{\boldsymbol{u}}
\newcommand{\bv}{\boldsymbol{v}}
\newcommand{\bw}{\boldsymbol{w}}
\newcommand{\bx}{\boldsymbol{x}}
\newcommand{\by}{\boldsymbol{y}}
\newcommand{\bz}{\boldsymbol{z}}

\newcommand{\bH}{\boldsymbol{H}}

\newcommand{\bX}{\boldsymbol{X}}

\newcommand{\bV}{\boldsymbol{V}}
\newcommand{\bW}{\boldsymbol{W}}

\newcommand{\bbeta}{\boldsymbol{\beta}}
\newcommand{\bgamma}{\boldsymbol{\gamma}}

\newcommand{\bOmega}{\boldsymbol{\Omega}}
\newcommand{\bSigma}{\boldsymbol{\Sigma}}

\title{\bf Path Following and Empirical Bayes Model Selection for Sparse Regression}
\date{}
\author{
Hua Zhou \\
Department of Statistics \\
North Carolina State University \\
Raleigh, NC 27695-8203 \\
\texttt{hua\_zhou@ncsu.edu} \\
\and
Artin Armagan  \\
SAS Institute, Inc\\
Cary, NC 27513 \\
\texttt{artin.armagan@sas.com} \\
\and
David B. Dunson \\
Department of Statistical Science\\
Duke University\\
Durham, NC 27708 \\
\texttt{dunson@stat.duke.edu} \\
}

\begin{document}
\maketitle

\baselineskip=20pt

\begin{abstract}
In recent years, a rich variety of regularization procedures have been proposed for high dimensional regression problems.  However, tuning parameter choice and computational efficiency in ultra-high dimensional problems remain vexing issues.  The routine use of $\ell_1$ regularization is largely attributable to the computational efficiency of the LARS algorithm, but similar efficiency for better behaved penalties has remained elusive. In this article, we propose a highly efficient path following procedure for combination of any convex loss function and a broad class of penalties. From a Bayesian perspective, this algorithm rapidly yields maximum a posteriori estimates at different hyper-parameter values. To bypass the inefficiency and potential instability of cross validation, we propose an empirical Bayes procedure for rapidly choosing the optimal model and corresponding hyper-parameter value. This approach applies to any penalty that corresponds to a proper prior distribution on the regression coefficients. While we mainly focus on sparse estimation of generalized linear models, the method extends to more general regularizations such as polynomial trend filtering after reparameterization. The proposed algorithm scales efficiently to large $p$ and/or $n$. Solution paths of 10,000 dimensional examples are computed within one minute on a laptop for various generalized linear models (GLM). Operating characteristics are assessed through simulation studies and the methods are applied to several real data sets.

\vspace{.1in}
\vspace{.1in}
\end{abstract}

\section{Introduction}
\label{Introduction}

Sparse estimation via regularization has become a prominent research area over the last decade finding interest across a broad variety of disciplines. Much of the attention was brought by lasso regression \citep{Tibshirani96Lasso}, which is simply $\ell_1$ regularization, where $\ell_{\eta}$ is the $\eta$-norm of a vector for $\eta>0$. It was not until the introduction of the LARS algorithm \citep{EfronHastieIainTibshirani04LARS} that lasso became so routinely used. This popularity is attributable to the excellent computational performance of LARS, a variant of which obtains the whole solution path of the lasso at the cost of an ordinary least squares estimation. Unfortunately lasso has some serious disadvantages in terms of estimation bias and model selection behavior. There has been a rich literature analyzing lasso and its remedies \citep{Fu98BridgeLasso,KnightFu00LassoAsymptotics,FanLi01SCAD,YuanLin05EmpiricalBayes,ZouHastie05Enet,ZhaoYu06LassoConsistency,Zou06AdaptiveLasso,MeinshausenBuhlmann06GraphLasso,ZouLi08Onestep,ZhangHuang08LassoBias}.

Motivated by disadvantages of $\ell_1$ regularization, some non-convex penalties are proposed which, when designed properly, reduce bias in large signals while shrinking noise-like signals to zero \citep{FanLi01SCAD,Candes08ReWtL1,Friedman08GPS,Armagan09Bridge,Zhang10MCP,Armagan11Pareto}. However, non-convex regularization involves difficult non-convex optimization. As a convex loss function plus concave penalties is a difference of two convex functions, an iterative algorithm for estimation at a fixed regularization parameter can be constructed by the majorization-minimization principle \citep{Lange10NumAnalBook}. At each iteration, the penalty is replaced by the supporting hyperplane tangent at the current iterate. As the supporting hyperplane majorizes the concave penalty function, minimizing the convex loss plus the linear majorizing function (an $\ell_1$ regularization problem) produces the next iterate, which is guaranteed to decrease the original penalized objective function. Many existing algorithms for estimation with concave penalties fall into this category \citep{FanLi01SCAD,HunterLi05MM,ZouLi08Onestep,Candes08ReWtL1,Armagan11Pareto}. Although being numerically stable and easy to implement, their (often slow) convergence to a local mode makes their performance quite sensitive to starting values in settings where $p>>n$. Coordinate descent is another algorithm for optimization in sparse regression at a fixed tuning parameter value and has found success in ultra-high dimensional problems \citep{Friedman07Coordinate,WuLange08Coordinate,FriedmanHastieTibshirani10GLMCD,Mazumder11SparseNet}.  Nevertheless the optimization has to be performed at a large number of grid points, making the computation rather demanding compared to path following algorithms such as LARS. For both algorithms the choice of grid points is tricky. When there are too few, important events along the path are missed; when there are too many, computation becomes expensive.

The choice of tuning parameter is another important challenge in regularization problems. Cross-validation is widely used but incurs considerable computational costs. An attractive alternative is to select the tuning parameter according to a model selection criterion such as the Akaike information criterion (AIC) \citep{Akaike74AIC} or the Bayesian information criterion (BIC) \citep{Schwartz78BIC}. These criteria choose the tuning parameter minimizing the negative log-likelihood penalized by the model dimension. In shrinkage estimation, however, the degrees of freedom is often unclear. Intriguing work by \cite{WangLeng07Lasso} and \cite{Wang07SCADBIC} extend BIC to be used with shrinkage estimators. However, the meaning of BIC as an empirical Bayes procedure is lost in such extensions. \citet{ZhangHuang08LassoBias} study the properties of generalized information criterion (GIC) in a similar context.

In this paper, we address these two issues for the general regularization problem
\begin{align}
    \min_{\bbeta} \, f(\bbeta) + \sum_{j \in {\cal S}} P_\eta(|\beta_j|, \rho), \label{eqn:pen-obj}
\end{align}
where $f$ is a twice differentiable convex loss function, ${\cal S} \subset \{1,\ldots,p\}$ indicates the subset of coefficients being penalized (${\cal S}$ for shrink), and $P_\eta(t,\rho)$ is a scalar penalty function. Here $\rho$ is the penalty tuning parameter and $\eta$ represents possible parameter(s) for a penalty family. Allowing a general penalization set ${\cal S}$ increases the applicability of the method, as we will see in Section \ref{sec:gen-reg}. Throughout this article we assume the following regularity conditions on the penalty function $P_\eta(t,\rho)$: (i) symmetric about 0 in $t$, (ii) $P_\eta(0,\rho)>-\infty$, for all $\rho \ge 0$, (iii) monotone increasing in $\rho \ge 0$ for any fixed $t$, (iv) non-decreasing in $t \ge 0$ for any fixed $\rho$, (v) first two derivatives with respect to $t$ exist and are finite.

The generality of (\ref{eqn:pen-obj}) is two-fold. First, $f$ can be any convex loss function. For least squares problems, $f(\bbeta)=\|\by - \bX \bbeta\|_2^2/2$. For generalized linear models (GLMs), $f$ is the negative log-likelihood function. For Gaussian graphical models, $f(\bOmega) = - \log \det |\bOmega| + \text{tr}(\hat \bSigma \bOmega)$, where $\hat \bSigma$ is a sample covariance matrix and the parameter $\bOmega$ is a precision matrix. Secondly, most commonly used penalties satisfy the aforementioned assumptions. These include but are not limited to
\begin{enumerate}
\item Power family \citep{FrankFriedman93Bridge}
\begin{align*}
    P_\eta(|\beta|,\rho) &= \rho |\beta|^\eta, \hspace{.2in} \eta \in (0,2].
\end{align*}
Varying $\eta$ from 2 bridges best subset regression to lasso ($l_1$) \citep{Tibshirani96Lasso,ChenDonohoSaunders01BasisPursuit} to ridge ($l_2$) regression \citep{HoerlKennard70Ridge}.
\item Elastic net \citep{ZouHastie05Enet}
\begin{align*}
    P_\eta(|\beta|, \rho) &= \rho [(\eta-1) \beta^2/2 + (2-\eta) |\beta|], \hspace{.2in} \eta \in [1,2].
\end{align*}
Varying $\eta$ from 1 to 2 bridges lasso regression to ridge regression.
\item Log penalty \citep{Candes08ReWtL1,Armagan11Pareto}
\begin{align*}
    P_\eta(|\beta|, \rho) &= \rho \ln (\eta + |\beta|), \hspace{.2in} \eta>0.
\end{align*}
This penalty was called generalized elastic net in \citet{Friedman08GPS} and log penalty in \citet{Mazumder11SparseNet}. Such a penalty is induced by a generalized Pareto prior thresholded and folded at zero, with the oracle properties studied by \citet{Armagan11Pareto}.
\item Continuous log penalty \citep{Armagan11Pareto}
\begin{align*}
    P(|\beta|, \rho) &= \rho \ln (\sqrt{\rho} + |\beta|).
\end{align*}
This version of log penalty was designed to guarantee continuity of the solution path when the design matrix is scaled and orthogonal \citep{Armagan11Pareto}.
\item The SCAD penalty \citep{FanLi01SCAD} is defined via its partial derivative
\begin{align*}
    \frac{\partial}{\partial |\beta|} P_{\eta}(|\beta|, \rho) &= \rho \left\{ 1_{\{|\beta| \le \rho\}} + \frac{(\eta \rho - |\beta|)_+}{(\eta-1)\rho} 1_{\{|\beta| > \rho\}} \right\}, \hspace{.2in} \eta>2.
\end{align*}
Integration shows SCAD as a natural quadratic spline with knots at $\rho$ and $\eta \rho$
\begin{align}
    P_{\eta}(|\beta|, \rho) &= \begin{cases}
        \rho|\beta| & |\beta|<\rho  \\
        \rho^2 + \frac{\eta \rho(|\beta|-\rho)}{\eta-1} - \frac{\beta^2-\rho^2}{2(\eta-1)}   & |\beta| \in [\rho, \eta\rho]   \\
        \rho^2 (\eta+1)/2  & |\beta|>\eta\rho
    \end{cases}.    \label{eqn:SCAD-density}
\end{align}
For small signals $|\beta|<\rho$, it acts as lasso; for larger signals $|\beta|>\eta \rho$, the penalty flattens and leads to the unbiasedness of the regularized estimate.
\item Similar to SCAD is the MC+ penalty \citep{Zhang10MCP} defined by the partial derivative
\begin{align*}
    \frac{\partial}{\partial |\beta|} P_{\eta}(|\beta|, \rho) &= \rho \left(1 - \frac{|\beta|}{\rho \eta} \right)_+.
\end{align*}
Integration shows that the penalty function
\begin{align}
    P_{\eta}(|\beta|, \rho) &= \left(\rho|\beta|-\frac{\beta^2}{2\eta} \right) 1_{\{|\beta|<\rho \eta\}} + \frac{\rho^2 \eta}{2} 1_{\{|\beta| \ge \rho \eta\}}, \hspace{.2in} \eta>0,    \label{eqn:MCP-density}
\end{align}
is quadratic on $[0,\rho\eta]$ and flattens beyond $\rho \eta$.
Varying $\eta$ from 0 to $\infty$ bridges hard thresholding ($\ell_0$ regression) to lasso ($\ell_1$) shrinkage.
\end{enumerate}
The derivatives of penalty functions will be frequently used for the development of the path algorithm and model selection procedure. They are listed in Table \ref{table:pen} of Supplementary Materials for convenience.

Our contributions are summarized as follows:
\begin{enumerate}
\item We propose a general path seeking strategy for the sparse regression framework (\ref{eqn:pen-obj}). To the best of our knowledge, no previous work exists at this generality, except the generalized path seeking (GPS) algorithm proposed in unpublished work by \cite{Friedman08GPS}. Some problems with the GPS algorithm motivated us to develop a more rigorous algorithm, which is fundamentally different from GPS.  Path algorithms for some specific combinations of loss and penalty functions have been studied before. Homotopy \citep{OsbornePresnellTurlach00LassoAlgo} and a variant of LARS \citep{EfronHastieIainTibshirani04LARS} compute the piecewise linear solution path of $\ell_1$ penalized linear regression efficiently. \citet{ParkHastie07GLMLasso} proposed an approximate path algorithm for $\ell_1$ penalized GLMs. A similar problem was considered by \citet{Wu10ODELasso} who devises a LARS algorithm for GLMs based on ordinary differential equations (ODEs). The ODE approach naturally fits problems with piecewise smooth solution paths and is the strategy we adopt in this paper. All of the aforementioned work deals with $\ell_1$ regularization which leads to convex optimization problems. Moving from convex to non-convex penalties improves the quality of the estimates but imposes great difficulties in computation. The PLUS path algorithm of \citep{Zhang10MCP} is able to track all local minima; however, it is specifically designed for the least squares problem with an MC+ penalty and does not generalize to (\ref{eqn:pen-obj}).

\item We propose an empirical Bayes procedure for the selection of a good model and the implied hyper/tuning parameters along the solution path. This method applies to any likelihood model with a penalty that corresponds to a proper shrinkage prior in the Bayesian setting. We illustrate the method with the power family (bridge) and log penalties which are induced by the exponential power and generalized double Pareto priors respectively. The regularization procedure resulting from the corresponding penalties is utilized as a model-search engine where each model and estimate along the path is appropriately evaluated by a criterion emerging from the prior used. \citet{YuanLin05EmpiricalBayes} took a somewhat similar approach in the limited setting of $\ell_1$ penalized linear regression.

\item The proposed path algorithm and empirical Bayes model selection procedure extend to a large class of generalized regularization problems such as polynomial trend filtering. Path algorithms for generalized $\ell_1$ regularization was recently studied by \citet{TibshiraniTaylor10GenLasso} and \citet{ZhouLange11LSPath} for linear regression and by \citet{ZhouWu11EPSODE} for general convex loss functions. Using non-convex penalties in these general regularization problems produces more parsimonious and  less biased estimates. Re-parameterization reformulates these problems as in
 (\ref{eqn:pen-obj}) which is solved by our efficient path algorithm.

\item A {\sc Matlab} toolbox for sparse regression is released on the first author's web site. The code for all examples in this paper is available on the same web site, observing the principle of reproducible research.
\end{enumerate}

The remainder of the article is organized as follows. The path following algorithm is derived in Section \ref{sec:path-algo}. The empirical Bayes criterion is developed in Section \ref{sec:emp-bayes} and illustrated for the power family and log penalties. Section \ref{sec:gen-reg} discusses extensions to generalized regularization problems. Various numerical examples are presented in Section \ref{sec:examples}. Finally we conclude with a discussion and future directions.

\section{Path Following for Sparse Regressions}
\label{sec:path-algo}

For a parameter vector $\bbeta \in \mathbb{R}^p$, we use ${\cal S}_{\cal Z}(\bbeta)=\{j \in {\cal S}: \beta_j=0\}$ to denote the set of penalized parameters that are zero and correspondingly ${\cal S}_{\bar {\cal Z}}(\bbeta) = \{j \in {\cal S}: \beta_j \ne 0\}$ is the set of nonzero penalized parameters. ${\cal A} = \bar {\cal S} \cup {\cal S}_{\bar {\cal Z}}$ indexes the current active predictors with unpenalized or nonzero penalized coefficients. It is convenient to define the Hessian,  $\bH_{\cal A}(\bbeta, \rho) \in \mathbb{R}^{|{\cal A}| \times |{\cal A}|}$, of the penalized objective function (\ref{eqn:pen-obj}) restricted to the active predictors with entries
\begin{align}
     H_{jk} &= \begin{cases}
     [d^2f(\bbeta)]_{jk} & j \in \bar {\cal S} \\
     [d^2f(\bbeta)]_{jk} + \frac{\partial^2P(|\beta_j|,\rho)}{\partial|\beta_j|^2} 1_{\{j=k\}} & j \in {\cal S}_{\bar {\cal Z}}   
     \end{cases}.   \label{eqn:HA}
\end{align}
Our path following algorithm revolves around the necessary condition for a local minimum. We denote the penalized objective function (\ref{eqn:pen-obj}) by $h(\bbeta)$ throughout this article.
\begin{lemma}[Necessary optimality condition]
If $\bbeta$ is a local minimum of (\ref{eqn:pen-obj}) at tuning parameter value $\rho$, then $\bbeta$ satisfies the stationarity condition
\begin{align}
    \nabla_j f(\bbeta) + \frac{\partial P(|\beta_j|,\rho)}{\partial |\beta_j|} \omega_j 1_{\{j \in {\cal S}\}} = 0, \hspace{.2in} j=1,\ldots,p,   \label{eqn:stationarity}
\end{align}
where the coefficients $\omega_j$ satisfy
\begin{align*}
   \omega_j \in \begin{cases}
   \{-1\} & \beta_j < 0 \\
   [-1, 1] & \beta_j = 0    \\
   \{1\} & \beta_j > 0
   \end{cases}.
\end{align*}
Furthermore, $\bH_{\cal A}(\bbeta,\rho)$ is positive semidefinite.
\end{lemma}
\begin{proof}
When the penalty function $P$ is convex, this is simply the first order optimality condition for unconstrained convex minimization \citep[Theorem 3.5]{Ruszczynski06Book}. When $P$ is non-convex, we consider the optimality condition coordinate-wise. For $j\in\{j:\beta_j \ne 0\}$, this is trivial. When $\beta_j=0$, $\beta_j$ being a local minimum implies that the two directional derivatives are non-negative. Then
\begin{align*}
    d_{\be_j} h(\bbeta) &= \lim_{t \downarrow 0} \frac{h(\bbeta+t \be_j)-h(\bbeta)}{t} = \nabla_j f(\bbeta) + \frac{\partial P(|\beta_j|,\rho)}{\partial |\beta_j|} \ge 0    \\
    d_{-\be_j} h(\bbeta) &= \lim_{t \uparrow 0} \frac{h(\bbeta+t \be_j)-h(\bbeta)}{t} = - \nabla_j f(\bbeta) + \frac{\partial P(|\beta_j|,\rho)}{\partial |\beta_j|} \ge 0,
\end{align*}
which is equivalent to (\ref{eqn:stationarity}) with $\omega_j \in [-1,1]$. Positive semidefiniteness of $\bH_{\cal A}$ follows from the second order necessary optimality condition when restricted to coordinates in ${\cal A}$.
\end{proof}
\noindent
We call any $\bbeta$ satisfying (\ref{eqn:stationarity}) a \emph{stationary point} at $\rho$. Our path algorithm tracks a stationary point along the path while sliding $\rho$ from infinity towards zero. When the penalized objective function $h$ is convex, e.g., $\eta \in [1,2]$ regime of the power family, elastic net, or $d^2h$ is positive semidefinite, the stationarity condition (\ref{eqn:stationarity}) is sufficient for a global minimum. When $h$ is not convex, the minimization problem is both non-smooth and non-convex and there lacks an easy-to-check sufficient condition for optimality. The most we can claim is that the directional derivatives at any stationary point are non-negative.
\begin{lemma}
\label{lemma:dir-deriv}
Suppose $\bbeta$ satisfies the stationarity condition (\ref{eqn:stationarity}). Then all directional derivatives at $\bbeta$ are non-negative, i.e.,
\begin{eqnarray}
    d_{\bv} h(\bbeta) = \lim_{t \downarrow 0} \frac{h(\bbeta+t\bv)-h(\bbeta)}{t} \ge 0    \label{eqn:dir-derivative}
\end{eqnarray}
for any $\bv \in \mathbb{R}^p$. Furthermore, if the penalized objective function $h$ is convex, then $\bbeta$ is a global minimum.
\end{lemma}
\begin{proof}
By definition of directional derivative and the stationarity condition (\ref{eqn:stationarity}),
\begin{eqnarray*}
    d_{\bv} h(\bbeta) &=& d_{\bv} f(\bbeta) + \sum_{j \in {\cal S}: \beta_j \ne 0} v_j \left. \frac{\partial P_\eta(t,\rho)}{\partial t} \right|_{t=|\beta_j|} \mathrm{sgn}(\beta_j) + \sum_{j \in {\cal S}: \beta_j = 0} |v_j| \left. \frac{\partial P_\eta(t,\rho)}{\partial t} \right|_{t=0}   \\
    &=& \sum_j v_j \nabla_j f(\bbeta) + \sum_{j \in {\cal S}: \beta_j \ne 0} v_j \left. \frac{\partial P_\eta(t,\rho)}{\partial t} \right|_{t=|\beta_j|} \mathrm{sgn}(\beta_j) + \sum_{j \in {\cal S}: \beta_j = 0} |v_j| \left. \frac{\partial P_\eta(t,\rho)}{\partial t} \right|_{t=0}   \\
    &=& \sum_{j \notin {\cal A}} |v_j| \left( \mathrm{sgn}(v_j) \cdot \nabla_j f(\bbeta) + \left. \frac{\partial P_\eta(t,\rho)}{\partial t} \right|_{t=0} \right)   \\
    &\ge& 0.
\end{eqnarray*}
Consider the scalar function $g(t) = h(\bbeta+t \bv)$. Convexity of $h$ implies that $g$ is convex too. Then the chord $
[g(t)-g(0)]/t = [h(\bbeta+t\bv) - h(\bbeta)]/t$ is increasing for $t \ge 0$. Thus $h(\bbeta+\bv) - h(\bbeta) \ge d_{\bv} h(\bbeta) \ge 0$ for all $\bv$, verifying that $\bbeta$ is a global minimum.
\end{proof}

We remark that, without convexity, non-negativeness of all directional derivatives does {\it not} guarantee local minimality, as demonstrated in the following example \citep[Exercise 1.12]{Lange04Optm}. Consider the bivariate function $f(x,y) = (y-x^2)(y-2x^2)$. Any directional derivative at origin (0,0) is non-negative since $\lim_{t \to 0} [f(ht,kt) - f(0,0)]/t = 0 $ for any $h,k \in \mathbb{R}$ and indeed, (0,0) is a local minimum along any line passing through it. However (0,0) is not a local minimum for $f$ as it is easy to see that $f(t,ct^2)<0$ for any $1<c<2$, $t \ne 0$, and that $f(t,ct^2)>0$ for any $c<1$ or $c>2$, $t \ne 0$. Figure \ref{fig:func} demonstrates how we go down hill along the parabola $y = 1.4x^2$ as we move away from $(0,0)$. In this article, we abuse terminology by the use of \emph{solution path} and in fact mean \emph{path of stationarity points}.
\begin{figure}[ht!]
\centering\includegraphics[width=4.5in]{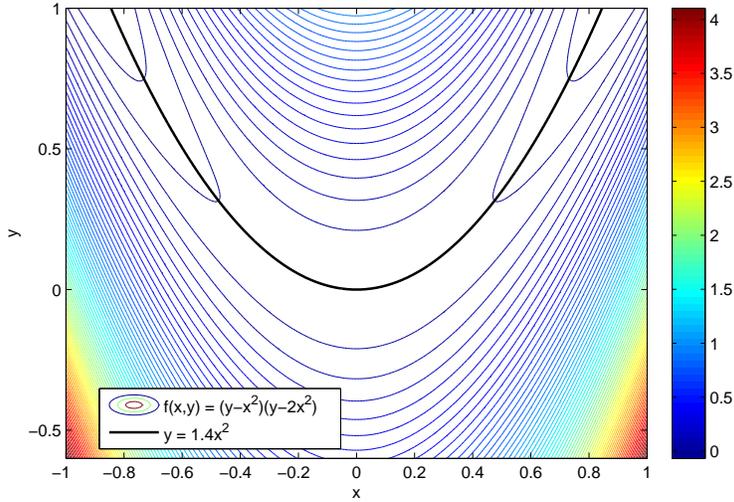}
\caption{Contours of $f(x,y)=(y-x^2)(y-2x^2)$ and the parabola $y = 1.4x^2$ that passes through $(0,0)$, which is not a local minimum although all directional derivatives at (0,0) are nonnegative.}
\label{fig:func}
\end{figure}

\subsection{Least squares with orthogonal design}
\label{subsec:thresholding}

Before deriving the path algorithm for the general sparse regression problem (\ref{eqn:pen-obj}), we first investigate a simple case: linear regression with orthogonal design, i.e., $f(\bbeta) = \|\by - \bX \bbeta\|_2^2/2$ where $\bx_j^t \bx_k = \|\bx_j\|_2^2 1_{\{j=k\}}$. This serves two purposes. First it illustrates the difficulties (discontinuities, local minima) of path seeking with non-convex penalties. Secondly, the thresholding operator for orthogonal design is the building block of the coordinate descent algorithm \citep{Friedman07Coordinate,WuLange08Coordinate,FriedmanHastieTibshirani10GLMCD,Mazumder11SparseNet} or iterative thresholding, which we rely on to detect the discontinuities in path following for the non-convex case.

For linear regression with orthogonal design, the penalized objective function in (\ref{eqn:pen-obj}) can be written in a component-wise fashion and the path solution is
\begin{align}
    \hat \beta_j(\rho) = \text{argmin}_{\beta} \, \frac{a_j}{2} (\beta - b_j)^2 + P_\eta(|\beta|,\rho) \label{eqn:thresholding}
\end{align}
where $a_j = \bx_j^t \bx_j = \|\bx_j\|_2^2$ and $b_j = \bx_j^t \by / \bx_j^t \bx_j$. The solution to (\ref{eqn:thresholding}) for some popular penalties is listed in Supplementary Materials. Similar derivations can be found in \citep{Mazumder11SparseNet}. Existing literature mostly assumes that the design matrix is standardized, i.e., $a_j=1$. As we argue in forthcoming examples, in many applications it is prudent not to do so. Figure \ref{fig:DP-illustration} depicts the evolution of the penalized objective function with varying $\rho$ and the solution path for $a_j=b_j=1$ and the log penalty $P_\eta(|\beta|,\rho) = \rho \ln (|\beta|+\eta)$ with $\eta=0.1$. At $\rho=0.2191$, the path solution jumps from a local minimum 0 to the other positive local minimum.

For the least squares problem with a non-orthogonal design, the coordinate descent algorithm iteratively updates $\beta_j$ by (\ref{eqn:thresholding}). When updating $\beta_j$ keeping other predictors fixed, the objective function takes the same format with $a_j = \|\bx_j\|_2^2$ and $b_j=\bx_j^t(\by-\bX_{-j}\bbeta_{-j})/\bx_j^t\bx_j$, where $\bX_{-j}$ and $\bbeta_{-j}$ denote the design matrix and regression coefficient vector without the $j$-th covariate. For a general twice differentiable loss function $f$, we approximate the smooth part $f$ by its Taylor expansion around current iterate $\bbeta^{(t)}$
\begin{align*}
    f(\bbeta) \approx f(\bbeta^{(t)}) + d f(\bbeta^{(t)}) (\bbeta-\bbeta^{(t)}) + \frac 12 (\bbeta-\bbeta^{(t)})^t d^2f(\bbeta^{(t)}) (\bbeta-\bbeta^{(t)})
\end{align*}
and then apply thresholding formula (\ref{eqn:thresholding}) for $\beta_j$ sequentially to obtain the next iterate $\bbeta^{(t+1)}$.

\begin{figure}
\begin{center}
$$
\begin{array}{cc}
\includegraphics[width=2.75in]{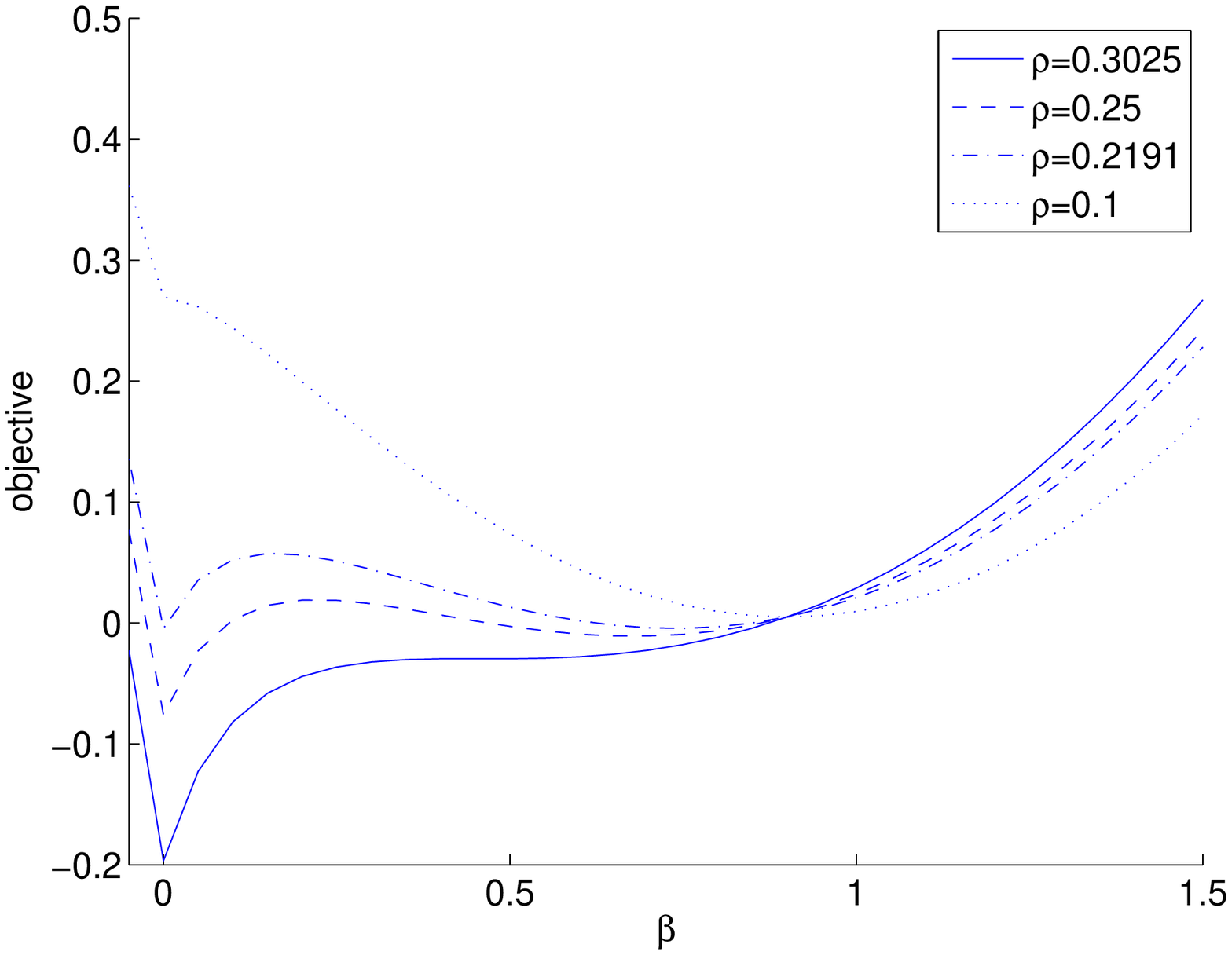} & \includegraphics[width=2.75in]{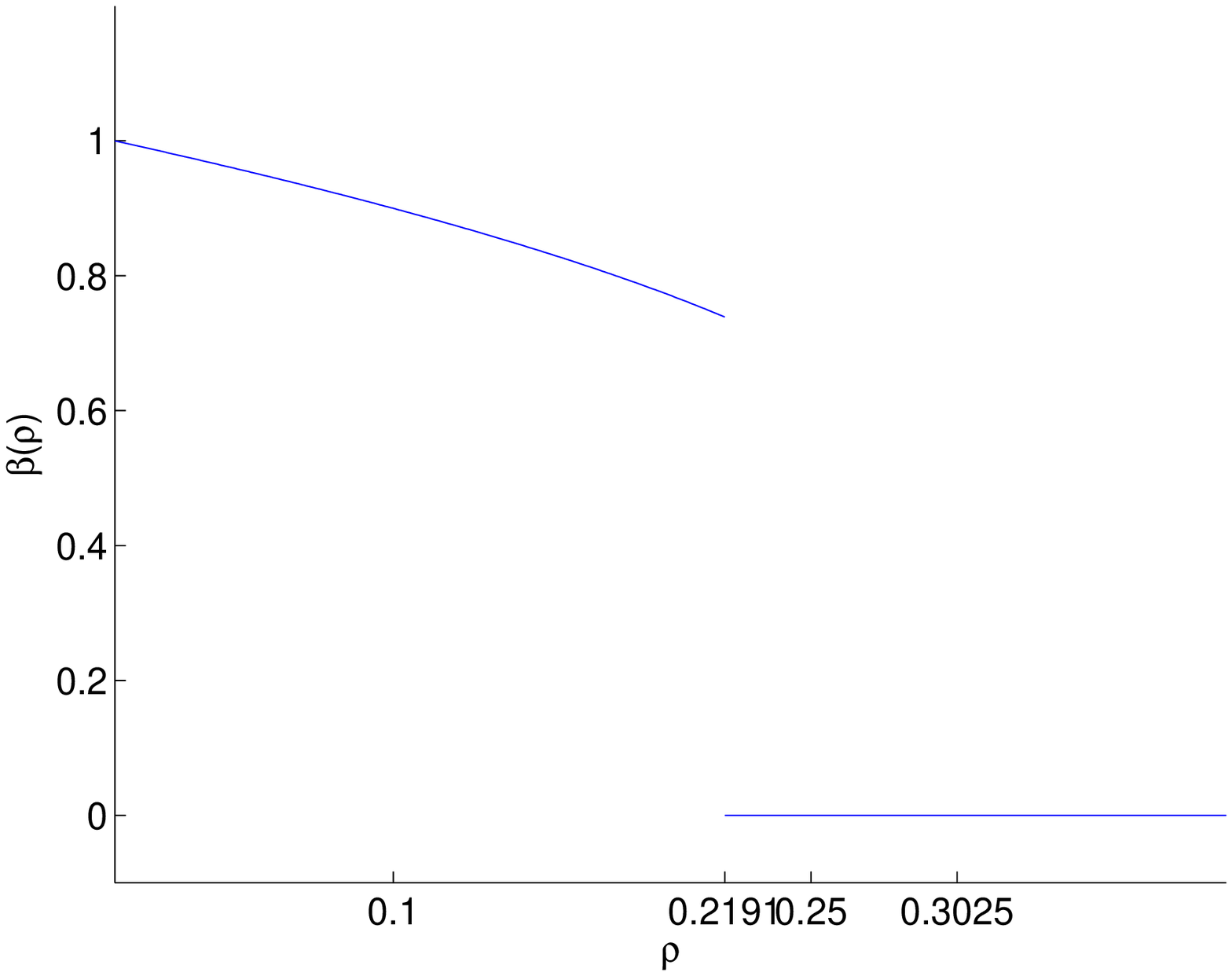}
\end{array}
$$
\end{center}
\caption{Log penalty with orthogonal design: $a=1$, $b=1$, $\eta=0.1$. Left: Graphs of the penalized objective function $a(\beta-b)^2/2 + \rho \ln (\eta+|\beta|)$ at different $\rho$. Right: Solution path. Here $ab\eta = 0.1$ and $a(\eta+|b|)^2/4=0.3025$. Discontinuity occurs somewhere between these two numbers.}
\label{fig:DP-illustration}
\end{figure}

\subsection{Path following for the general case}

The preceding discussion illustrates two difficulties with path seeking in sparse regression using non-convex penalties. First the solution path may not be continuous. Discontinuities occur when predictors enter the model at a non-zero magnitude or vice versa. This is caused by jumping between local minima. Second, in contrast to lasso, the solution path is no longer piecewise linear. This prohibits making giant jumps along the path like LARS.

One strategy is to optimize (\ref{eqn:pen-obj}) at a grid of penalty intensity $\rho$ using coordinate descent. This has found great success with lasso and elastic net regression \citep{Friedman07Coordinate,WuLange08Coordinate,FriedmanHastieTibshirani10GLMCD}. The recent article \citep{Mazumder11SparseNet} explores the coordinate descent strategy with non-convex penalties. In principle it involves applying the thresholding formula to individual regression coefficients until convergence. However, determining the grid size for tuning parameter $\rho$ in advance could be tricky. The larger the grid size the more likely we are to miss important events along the path, while the smaller the grid size the higher the computational costs.

In this section we devise a path seeking strategy that tracks the solution path smoothly while allowing abrupt jumping (due to discontinuities) between segments. The key observation is that each path segment is smooth and satisfies a simple ordinary differential equation (ODE). Recall that the active set ${\cal A} = \bar {\cal S} \cup {\cal S}_{\bar {\cal Z}}$ indexes all unpenalized and nonzero penalized coefficients and $\bH_{\cal A}$ is the Hessian of the penalized objective function restricted to parameters in ${\cal A}$.
\begin{proposition} \label{prop:ode}
The solution path $\bbeta(\rho)$ is continuous and differentiable at $\rho$ if $\bH_{\cal A}(\bbeta,\rho)$ is positive definite. Moreover the solution vector $\bbeta(\rho)$ satisfies
\begin{align}
    \frac{d \bbeta_{\cal A}(\rho)}{d \rho} &= - \bH_{\cal A}^{-1}(\bbeta,\rho) \cdot \bu_{\cal A}(\bbeta,\rho), \label{eqn:sol-ode-compact}
\end{align}
where the matrix $\bH_{\cal A}$ is defined by (\ref{eqn:HA}) and the vector $\bu_{\cal A}(\bbeta)$ has entries
\begin{align*}
    u_j(\bbeta,\rho) = \begin{cases}
    \frac{\partial^2 P_\eta(|\beta_j|,\rho)}{\partial |\beta_j| \partial \rho} \text{sgn}(\beta_j) & j \in {\cal S}_{\bar {\cal Z}} \\
     0 & j \in \bar {\cal S}
     \end{cases}.
\end{align*}
\end{proposition}

\begin{proof}
Write the stationarity condition (\ref{eqn:stationarity}) for active predictors as a vector equation $k(\bbeta_{\cal A},\rho) = {\bf 0}$. To solve for $\bbeta_{\cal A}$ in terms of $\rho$, we apply the implicit function theorem \citep{Lange04Optm}. This requires calculating the differential of $k$ with respect to the dependent variables $\bbeta_{\cal A}$ and the independent variable $\rho$
\begin{align*}
    \partial_{\bbeta_{\cal A}} k(\bbeta_{\cal A},\rho) &= \bH_{\cal A}(\bbeta,\rho) \\
    \partial_{\rho} k(\bbeta_{\cal A},\rho) &= \bu_{\cal A}(\bbeta).
\end{align*}
Given the non-singularity of $\bH_{\cal A}(\bbeta,\rho)$, the implicit function theorem applies and shows the continuity and differentiability of $\bbeta_{\cal A}(\rho)$ at $\rho$. Furthermore, it supplies the derivative (\ref{eqn:sol-ode-compact}).
\end{proof}
\noindent
Proposition \ref{prop:ode} suggests that solving the simple ODE segment by segment is a promising path following strategy. However, the potential discontinuity along the path caused by the non-convex penalty has to be taken care of. Note that the stationarity condition (\ref{eqn:stationarity}) for inactive predictors implies
\begin{align*}
    \omega_j = - \frac{\nabla_jf(\bbeta)}{\frac{\partial}{\partial |\beta|}P(|\beta|,\rho)}, \hspace{.2in} j \in {\cal S}_{\cal Z},
\end{align*}
and provides one indicator when the coefficient should escape to ${\cal S}_{\bar {\cal Z}}$ during path following. However, due to the discontinuity, a regression coefficient $\beta_j$, $j \in {\cal S}_{\cal Z}$, may escape with $\omega_j$ in the interior of (-1,1). A more reliable implementation should check whether an inactive regression coefficient $\beta_j$ becomes nonzero using the thresholding formulae at each step of path following. Another complication that discontinuity causes is that occasionally the active set ${\cal S}_{\cal Z}$ may change abruptly along the path especially when predictors are highly correlated. Therefore whenever a discontinuity is detected, it is advisable to use any nonsmooth optimizer, e.g., coordinate descent, to figure out the set configuration and starting point for the next segment. We pick up coordinate descent due to its simple implementation. Our path following strategy is summarized in Algorithm \ref{algo:path}.

\begin{algorithm}
\begin{algorithmic}
\STATE Determine the first penalized predictor $j^*$ to enter model and the corresponding $\rho_{\text{max}}$
\STATE Initialize ${\cal S}_{\cal Z} = \{j^*\}$, ${\cal S}_{\bar {\cal Z}} = {\cal S} \setminus \{j^*\}$, and $\bbeta(\rho_{\text{max}}) = \text{argmin}_{\bbeta_{\cal S}={\bf 0}} f(\bbeta)$
\REPEAT
\STATE Solve ODE
$$
\frac{d \bbeta_{\cal A}(\rho)}{d\rho} = - \bH_{\cal A}(\bbeta,\rho)^{-1} \bu_{\cal A}(\bbeta,\rho)
$$
until (1) an active penalized predictor $\beta_j$, $j \in {\cal S}_{\bar {\cal Z}}$, becomes 0, or (2) an inactive penalized coefficient $w_j$, $j \in {\cal S}_{\cal Z}$, hits 1 or -1, or (3)
an inactive penalized predictor $\beta_j$, $j \in {\cal S}_{\cal Z}$, jumps from 0 to a nonzero minimum, or (4) the matrix $\bH_{\cal A}(\bbeta,\rho)$ becomes singular.
\IF {(1) or (2)}
\STATE Update sets $S_{\cal Z}$ and $S_{\bar {\cal Z}}$
\ELSE
\STATE Apply coordinate descent at current $\rho$ to determine $S_{\cal Z}$, $S_{\bar {\cal Z}}$ and $\beta_{\cal A}$ for next segment
\ENDIF
\UNTIL{termination criterion is met}
\end{algorithmic}
\caption{Path following for sparse regression.}
\label{algo:path}
\end{algorithm}

Several remarks on Algorithm \ref{algo:path} are relevant here.
\begin{remark}[Path following direction]
The ODE (\ref{eqn:sol-ode-compact}) is written in the usual sense and gives the derivative as $\rho$ increases. In sparse regression, we solve in the reverse direction and shall take the opposite sign.
\end{remark}

\begin{remark}[Termination criterion]
Termination criterion for path following may depend on the specific likelihood model. For linear regression, path seeking stops when the number of active predictors exceeds the rank of design matrix $|{\cal A}|>\mathrm{rank}(\bX)$. The situation is more subtle for logistic or Poisson log-linear models due to separation. In binary logistic regression, complete separation occurs when there exists a vector $\bz \in \mathbb{R}^p$ such that $\bx_i^t \bz>0$ for all $y_i=1$ and $\bx_i^t \bz<0$ for all $y_i=0$. When complete separation happens, the log-likelihood is unbounded and the MLE occurs at infinity along the direction $\bz$. The log-likelihood surface behaves linearly along this direction and dominates many non-convex penalties such as power, log, MC+, and SCAD, which is almost flat at infinity. This implies that the penalized estimate also occurs at infinity. Path seeking should terminate whenever separation is detected, which may happen when $|{\cal A}|$ is much smaller than the rank of the design matrix in large $p$ small $n$ problems. Separation occurs in the Poisson log-linear model too. Our implementation also allows users to input the maximum number of selected predictors until path seeking stops, which is convenient for exploratory analysis of ultra-high dimensional data.
\end{remark}

\begin{remark}[Computational Complexity and Implementation]
Any ODE solver repeatedly evaluates the derivative \eqref{eqn:sol-ode-compact}. The path segment stopping events (1)-(4) are checked during each derivative evaluation. Since the Hessian restricted to the active predictors is always positive semidefinite and the inactive penalized predictors are checked by thresholding, the quality of solution along the path is as good as any fixed tuning parameter optimizer such as coordinate descent \citep{Mazumder11SparseNet}. Computational complexity of Algorithm \ref{algo:path} depends on the loss function, number of smooth path segments, and the method for solving the ODE. Evaluating derivative \eqref{eqn:sol-ode-compact} takes $O(n|{\cal A}|^2)$ flops for calculating the Hessian of a GLM loss $\ell$ and takes $O(|{\cal A}|^3)$ flops for solving the linear system. Detecting jumps of inactive penalized predictor by thresholding takes $O(|\tilde {\cal A}|)$ flops. The cost of $O(|{\cal A}|^3)$ per gradient evaluation is not as daunting as it appears. Suppose any fixed tuning parameter optimizer is utilized for path following with a warm start. When at a new $\rho$, assuming that the active set ${\cal A}$ is known, even the fastest Newton's method needs to solve the same linear system multiple times until convergence. The efficiency of Algorithm \ref{algo:path} lies in the fact that no iterations are needed at any $\rho$ and it adaptively chooses step sizes to catch all events along the path. Algorithm \ref{algo:path} is extremely simple to implement using software with a reliable ODE solver such as the {\tt ode45} function in {\sc Matlab} and the {\tt deSolve} package in {\sc R} \citep{Soetaert10RdeSolve}. For instance, the rich numerical resources of {\sc Matlab} include differential equation solvers that alert the user when certain events such as those stopping rules in Algorithm \ref{algo:path} are fulfilled.
\end{remark}

\begin{remark}[Knots in SCAD and MC+]
Solving the ODE (\ref{eqn:sol-ode-compact}) requires the second order partial derivatives $\frac{\partial^2}{\partial |\beta|^2} P(|\beta|,\rho)$ and $\frac{\partial^2}{\partial |\beta| \partial \rho} P(|\beta|,\rho)$ of the penalty functions, which are listed in Table \ref{table:pen}. Due to their designs, these partial derivatives are undetermined for SCAD and MC+ penalties at the knots: $\{\rho,\eta\rho\}$ for SCAD and $\{\eta \rho\}$ for MC+. However only the directional derivatives are needed, which are well-defined. Specifically we use $\frac{\partial^2}{\partial |\beta|^2} P(|\beta|,\rho_-)$ and $\frac{\partial^2}{\partial |\beta| \partial \rho} P(|\beta|,\rho_-)$. In practice, the ODE solver rarely steps on these knots exactly due to numerical precision.
\end{remark}

Finally, switching the role of $\rho$ and $\eta$, the same argument leads to an analogous result for path following in the penalty parameter $\eta$ with a fixed regularization parameter $\rho$. In this article we focus on path following in $\rho$ with fixed $\eta$ in the usual sense. Implications of the next result will be investigated in future work.
\begin{proposition}[Path following in $\eta$]
\label{prop:ode}
Suppose the partial derivative $\frac{\partial P_\eta(t,\rho)}{\partial t \partial \eta}$ exists at all $t>0$ and $\rho$. For fixed $\rho$, the solution path $\bbeta(\eta)$ is continuous and differentiable at $\eta$ if $\bH_{\cal A}(\bbeta,\eta)$ is positive definite. Moreover the solution vector $\bbeta(\eta)$ satisfies
\begin{eqnarray*}
    \frac{d \bbeta_{\cal A}(\eta)}{d \eta} &=& - \bH_{\cal A}^{-1}(\bbeta,\eta) \cdot \bu_{\cal A}(\bbeta,\eta),
\end{eqnarray*}
where the matrix $\bH_{\cal A}$ is defined by (\ref{eqn:HA}) and the vector $\bu_{\cal A}(\bbeta)$ has entries
\begin{align*}
    u_j(\bbeta,\eta) = \begin{cases}
    \frac{\partial^2 P_\eta(|\beta_j|,\rho)}{\partial |\beta_j| \partial \eta} \text{sgn}(\beta_j) & j \in {\cal S}_{\bar {\cal Z}} \\
     0 & j \in \bar {\cal S}
     \end{cases}.
\end{align*}
\end{proposition}

\section{Empirical Bayes Model Selection}
\label{sec:emp-bayes}

In practice, the regularization parameter $\rho$ in sparse regression is tuned according to certain criteria. Often we wish to avoid cross-validation and rely on more efficient procedures. AIC, BIC and similar variants have frequently been used. Recall that BIC arises from a Laplace approximation to the log-marginal density of the observations under a Bayesian model. The priors on the parameters are specifically chosen to be normal with mean set at the maximum likelihood estimator and covariance that conveys the Fisher information observed from one observation. This allows for a rather diffuse prior relative to the likelihood. Hence the resulting maximum a posteriori estimate is the maximum likelihood estimator. Often users plug in the estimates from sparse regression into AIC or BIC to assess the quality of the estimate/model. In this section we derive an appropriate empirical Bayes criterion that corresponds to the exact prior under which we are operating. All necessary components for calculating the empirical Bayes criterion fall out nicely from the path following algorithm. A somewhat similar approach was taken by \cite{YuanLin05EmpiricalBayes} to pick an appropriate tuning parameter for the lasso penalized least squares noting that the underlying Bayesian model is formed with a mixture prior -- a spike at zero and a double exponential distribution on $\beta_j \in \mathbb{R}$.

Conforming to previous notation, a model is represented by the active set ${\cal A} = \bar {\cal S} \cup {\cal S}_{\cal Z}$ which includes both un-penalized and selected penalized regression coefficients. By Bayes formula, the probability of a model ${\cal A}$ given data $\by$ is
\begin{eqnarray*}
    p({\cal A}|\by) = \frac{p(\by|{\cal A})p({\cal A})}{p(\by)}.
\end{eqnarray*}
Assuming equal prior probability for all models, an appropriate Bayesian criterion for model comparison is the marginal data likelihood $p(\by|{\cal A})$ of model ${\cal A}$. If the penalty in the penalized regression is induced by a proper prior $\pi(\beta)$ on the regression coefficients, the marginal likelihood is calculated as
\begin{eqnarray}
    p(\by|{\cal A}) = \int \pi(\bbeta_{\cal A},\by) \, d\bbeta_{\cal A} = \int \pi(\by|\bbeta_{\cal A}) \prod_{j \in {\cal A}} \pi(\beta_j) \, d\bbeta_{\cal A}. \label{eqn:post-integral}
\end{eqnarray}
In most cases the integral cannot be analytically calculated. Fortunately the Laplace approximation is a viable choice in a similar manner to BIC, in which $\pi(\beta_j)$ is taken as the vaguely informative unit information prior. We illustrate this general procedure with the log and power penalties. Note that both the regularization parameter $\rho$ and penalty parameter $\eta$ are treated as hyper-parameters in priors. Thus the procedure not only allows comparison of models along the path with fixed $\eta$ but also models with distinct $\eta$.

\subsection{Log penalty}

The log penalty arises from a generalized double Pareto prior \citep{Armagan11Pareto} on regression coefficients
\begin{eqnarray*}
    \pi(\beta|\alpha,\eta) = \frac{\alpha \eta^\alpha}{2} (|\beta|+\eta)^{-(\alpha+1)}, \hspace{.2in} \alpha,\eta>0.
\end{eqnarray*}
Writing $\rho=\alpha+1$ and placing a generalized double Pareto prior on $\beta_j$ for active coefficients $j \in {\cal A}$ estimated at $(\rho,\eta)$, the un-normalized posterior is given by
\begin{eqnarray*}
\pi(\bbeta_{\cal A},\by|\rho,\eta) &=& \left\{\frac{(\rho-1)\eta^{\rho-1}}{2}\right\}^q \exp\left\{\ell(\bbeta_{\cal A})-\rho\sum_{j\in \mathcal{A}}\ln(\eta+|\beta_j|)\right\},
\nonumber \\
&=& \left\{\frac{(\rho-1)\eta^{\rho-1}}{2}\right\}^q \exp\left\{-h(\bbeta_{\cal A})\right\}, \label{post2}
\end{eqnarray*}
where $q = |{\cal A}|$ and $h(\bbeta_{\cal A}) =-\ell(\bbeta_{\cal A}) + \rho \sum_{j\in \mathcal{A}} \ln(\eta+|\beta_j|)$.
Then a Laplace approximation to the integral (\ref{eqn:post-integral}) enables us to assess the relative quality of an estimated model $\hat {\cal A}$ at particular hyper/tuning parameter values $(\rho,\eta)$ by $\mathrm{EB}(\rho,\eta) = - \ln p(\by|\tilde {\cal A})$. The following result displays the empirical Bayes criterion for the log penalty and then specializes to the least squares case which involves unknown variance. Note that, by the definition of double Pareto prior, $\rho>1$. Otherwise the prior on $\beta_j$ is no longer proper.

\begin{proposition}
For $\rho>1$, an empirical Bayes criterion for the log penalized regression is
\begin{eqnarray*}
\mathrm{EB}_{\log}(\rho,\eta) \equiv -q \ln\left\{\left(\frac{\pi}{2}\right)^{1/2}(\rho-1)\eta^{\rho-1}\right\} + h(\tilde \bbeta) + \frac{1}{2} \log \det \bH_{\cal A}(\tilde \bbeta),
\end{eqnarray*}
where $\tilde \bbeta$ is the path solution at $(\rho,\eta)$, ${\cal A} = {\cal A}(\tilde \bbeta)$ is the set of active regression coefficients, $q=|{\cal A}|$, and $\bH_{\cal A}$ is the restricted Hessian defined by (\ref{eqn:HA}). Under a linear model, it becomes
\begin{eqnarray*}
\mathrm{EB}_{\log}(\rho,\eta) \equiv -q \ln\left\{\left(\frac{\pi\tilde{\sigma}^2}{2}\right)^{1/2}\left(\frac{\rho}{\tilde{\sigma}^2}-1\right)\eta^{\rho/\tilde{\sigma}^2-1}\right\} + \frac{h(\tilde \bbeta)}{\tilde{\sigma}^2} + \frac{1}{2} \log \det \bH_{\cal A}(\tilde \bbeta),
\label{mlike1blin}
\end{eqnarray*}
where
\begin{equation*}
\tilde{\sigma}^2 = \mathrm{argmin}_{\sigma^2} \left\{-\frac{n-q}{2}\ln\sigma^2+\frac{q\rho}{\sigma^2}\ln\eta + q\ln\left(\frac{\rho}{\sigma^2}-1\right) - \frac{h(\tilde \bbeta)}{\sigma^2}\right\}.
\end{equation*}
\end{proposition}

\begin{proof}
The Laplace approximation to the normalizing constant (\ref{eqn:post-integral}) is given by
\begin{equation*}
\ln p(\by|\tilde {\cal A}) \approx \ln \pi(\tilde {\bbeta}_{\tilde {\cal A}},\by|\rho, \eta) + \frac{q}{2} \ln(2\pi) - \frac{1}{2} \log \det d^2h(\tilde{\bbeta}_{\tilde {\cal A}}),
\end{equation*}
where $\tilde {\bbeta}_{\tilde {\cal A}} = \mathrm{argmin}_{\bbeta_{\tilde {\cal A}}} h(\bbeta_{\tilde {\cal A}}) = \mathrm{argmin}_{\bbeta_{\tilde {\cal A}}} - \ell(\bbeta_{\tilde {\cal A}}) + \rho \sum_{j \in \tilde {\cal A}} \ln(\eta+|\beta_j|)$ and $[d^2 h(\bbeta_{\tilde {\cal A}})]_{jk} = [-d^2 \ell (\bbeta_{\tilde {\cal A}})]_{jk} + \rho (\eta+|\beta_j|)^{-2} 1_{\{j=k\}}$ for $j,k \in \mathcal{A}$. Then the empirical Bayes criterion is
\begin{eqnarray*}
\mathrm{EB}(\rho,\eta) &=& - \ln p(\by|\tilde {\cal A}) \\
&\approx& q\ln2 - q(\rho-1) \ln \eta - q\ln(\rho-1) + \min_{\bbeta_{\tilde {\cal A}}} h(\bbeta_{\tilde {\cal A}}) - \frac{q}{2} \ln(2\pi) + \frac{1}{2} \log \det d^{2} h(\tilde{\bbeta}_{\tilde {\cal A}})  \\
&=& - q \ln \left\{\left(\frac{\pi}{2}\right)^{1/2}(\rho-1)\eta^{\rho-1}\right\} + \min_{\bbeta_{\tilde {\cal A}}} h(\bbeta_{\tilde {\cal A}}) + \frac{1}{2} \log \det d^2 h(\tilde{\bbeta}_{\tilde {\cal A}}).
\end{eqnarray*}
Now consider the linear model with unknown variance $\sigma^2$,
\begin{eqnarray*}
    \pi(\bbeta_{\cal A},\by|\alpha,\eta,\sigma^2) &=& (2\pi\sigma^2)^{-n/2} \left( \frac{\alpha\eta^\alpha}{2}\right)^q \nonumber \\
    & &\times\exp\left\{-\frac{\|\by-\bX_{\cal A}\bbeta_{\cal A}\|_2^2 + 2\sigma^2(\alpha+1)\sum_{j\in \mathcal{A}} \ln(\eta+|\beta_j|)}{2\sigma^2}\right\} \\
    &=& (2\pi\sigma^2)^{-n/2}\left\{\frac{(\rho/\sigma^2-1)\eta^{\rho/\sigma^2-1}}{2}\right\}^q \exp \{ - h(\bbeta_{\cal A})/\sigma^2 \},
\end{eqnarray*}
where $\rho=\sigma^2(\alpha+1)$ and $h(\bbeta_{\cal A}) = \|\by-\bX_{\cal A}\bbeta_{\cal A}\|_2^2/2 + \rho \sum_{j\in \mathcal{A}} \ln(\eta+|\beta_j|)$.
The Laplace approximation to the normalizing constant is then given by
\begin{equation*}
\ln p(\by|\tilde {\cal A},\sigma^2) \approx \ln \pi(\tilde{\bbeta}_{\tilde {\cal A}},\by|\tilde {\cal A},\sigma^2)+\frac{q}{2}\ln(2\pi) - \frac{1}{2} \log \det [\sigma^{-2} d^2 h(\tilde{\bbeta}_{\tilde {\cal A}})],
\end{equation*}
which suggests the empirical Bayes criterion
\begin{eqnarray*}
\mathrm{EB}(\eta,\rho|\sigma^2) &\approx& \frac{n-q}{2} \ln(2\pi\sigma^2) + q\ln2 - q\left(\frac{\rho}{\sigma^2}-1\right) \ln\eta - q\ln\left(\frac{\rho}{\sigma^2}-1\right)\nonumber\\
& & + \frac{\min_{\bbeta_{\tilde {\cal A}}} h(\bbeta_{\tilde {\cal A}})}{\sigma^2} + \frac{1}{2} \log \det d^{2} h(\tilde{\bbeta}_{\tilde {\cal A}}).
\end{eqnarray*}
Given $\tilde {\cal A}$, we can easily compute the value $\sigma^2$ that minimizes the right-hand side
\begin{equation*}
\tilde{\sigma}^2 = \mathrm{argmin}_{\sigma^2} \left\{ \frac{n-q}{2} \ln \sigma^2 - \frac{q\rho}{\sigma^2} \ln \eta - q\ln\left(\frac{\rho}{\sigma^2}-1\right) + \frac{\min_{\bbeta_{\tilde {\cal A}}} h(\bbeta_{\tilde {\cal A}})}{\sigma^2}\right\}.
\end{equation*}
\end{proof}

\subsection{Power Family}

The power family penalty is induced by an exponential power prior on the regression coefficients
\begin{eqnarray*}
    \pi(\beta|\rho,\eta) = \frac{\eta\rho^{1/\eta}}{2\Gamma(1/\eta)} e^{-\rho |\beta|^\eta}, \hspace{.2in} \rho, \eta >0.
\end{eqnarray*}
The unnormalized posterior of the regression coefficients given a model $\cal A$ can be written as
\begin{eqnarray}
\pi(\bbeta_{\cal A},\by|\rho,\eta) &=& \left( \frac{\eta\rho^{1/\eta}}{2\Gamma(1/\eta)} \right)^q \exp\left\{\ell(\bbeta_{\cal A})-\rho\sum_{j \in \mathcal{A}}|\beta_j|^{\eta}\right\} \nonumber   \\
&=& \left( \frac{\eta\rho^{1/\eta}}{2\Gamma(1/\eta)} \right)^q \exp\left\{-h(\bbeta_{\cal A})\right\}
\label{post1pf}
\end{eqnarray}
Again the Laplace approximation to the posterior $p(\by|\hat {\cal A})$ yields the following empirical Bayes criterion for power family penalized regression.
\begin{proposition}
An empirical Bayes criterion for the power family penalized regression is
\begin{eqnarray*}
\mathrm{EB}_{\mathrm{PF}}(\rho,\eta) \equiv -q \ln\frac{\sqrt{\pi}\eta\rho^{1/\eta}}{\sqrt{2}\Gamma(1/\eta)} + h(\tilde \bbeta) + \frac{1}{2} \log \det \bH_{\cal A}(\tilde \bbeta).
\label{mlike2b}
\end{eqnarray*}
For linear regression, it becomes
\begin{eqnarray*}
\mathrm{EB}_{\mathrm{PF}}(\rho,\eta) \equiv -q \ln\frac{\sqrt{\pi}\eta\rho^{1/\eta}}{\sqrt{2}\Gamma(1/\eta)}+\left(\frac{n-q}{2}-\frac{q}{\eta}\right)\left\{1+\ln\frac{h(\tilde \bbeta)}{(n-q)/2+q/\eta}\right\} + \frac{1}{2} \log \det \bH_{\cal A}(\tilde \bbeta).
\label{mlike2bpflin}
\end{eqnarray*}
\end{proposition}

\begin{proof}
Given a certain model $\tilde {\cal A}$ observed at $(\rho,\eta)$, the Laplace approximation to the normalizing constant is given by
\begin{equation*}
\ln p(\by|\tilde {\cal A}) \approx \ln \pi(\tilde{\bbeta}_{\tilde {\cal A}},\by|\tilde {\cal A})+\frac{q}{2}\ln(2\pi) - \frac{1}{2} \log \det d^2 h(\tilde{\bbeta}_{\tilde {\cal A}})
\end{equation*}
where $\tilde{\bbeta}_{\tilde {\cal A}}=\mathrm{argmin}_{\bbeta_{\tilde {\cal A}}}h(\bbeta_{\tilde {\cal A}})=\mathrm{argmin}_{\bbeta_{\tilde {\cal A}}}-\ell(\bbeta_{\tilde {\cal A}})+\rho\sum_{j \in \mathcal{A}}|\beta_j|^{\eta}$ and $[d^2 h(\bbeta_{\tilde {\cal A}})]_{jk}=[-d^2 \ell(\bbeta_{\tilde {\cal A}})]_{jk}+\rho\eta(\eta-1)|\beta_j|^{\eta-2}1_{\{j=k\}}$ for $j,k \in \mathcal{A}$. Then
\begin{eqnarray*}
\ln p(\by|\tilde {\cal A}) &\approx& \frac{q}{2}\ln(2\pi)-q\ln2+q\ln\eta+\frac{q}{\eta}\ln\rho-q\ln\Gamma(1/\eta)\nonumber\\
& &-\min_{\bbeta_{\tilde {\cal A}}}h(\bbeta_{\tilde {\cal A}})-\frac{1}{2}\log \det d^{2}h(\tilde{\bbeta}_{\tilde {\cal A}}),
\end{eqnarray*}
which yields
\begin{equation*}
\mbox{EB}(\rho,\eta) \equiv -q\ln\frac{\sqrt{\pi}\eta\rho^{1/\eta}}{\sqrt{2}\Gamma(1/\eta)}+\min_{\bbeta_{\tilde {\cal A}}}h(\bbeta_{\tilde {\cal A}})+\frac{1}{2}\log \det d^{2}h(\tilde{\bbeta}_{\tilde {\cal A}}).
\end{equation*}
Now consider the linear model case,
\begin{eqnarray*}
    \pi(\bbeta_{\cal A},\by|\rho,\eta,\sigma^2) &=& (2\pi\sigma^2)^{-n/2}\left(\frac{\eta\rho^{1/\eta}}{2\sigma^{2/\eta}\Gamma(1/\eta)}\right)^q \exp\left\{- \frac{\|\by-\bX_{\cal A} \bbeta_{\cal A}\|_2^2/2 + \rho\sum_{j \in \mathcal{A}}|\beta_j|^{\eta}}{\sigma^2}\right\} \\
    &=& (2\pi\sigma^2)^{-n/2} \left( \frac{\eta\rho^{1/\eta}}{2\sigma^{2/\eta}\Gamma(1/\eta)} \right)^q \exp \left\{-h(\bbeta_{\cal A})/\sigma^2 \right\},
\end{eqnarray*}
where $h(\bbeta_{\cal A})= \|\by-\bX_{\cal A} \bbeta_{\cal A}\|_2^2/2 + \rho\sum_{j \in \mathcal{A}}|\beta_j|^{\eta}$.
The Laplace approximation to the normalizing constant at an estimated model $\tilde A$ is then given by
\begin{equation*}
    \ln p(\by|\tilde {\cal A},\sigma^2) \approx \ln \pi(\tilde{\bbeta}_{\tilde {\cal A}},\by|\tilde {\cal A},\sigma^2)+\frac{q}{2}\ln(2\pi)-\frac{1}{2}\log \det [\sigma^{-2} d^{2} h(\tilde{\bbeta}_{\tilde {\cal A}})],
\end{equation*}
where $\tilde{\bbeta}_{\tilde {\cal A}}=\mathrm{argmin}_{\bbeta_{\tilde {\cal A}}}h(\bbeta_{\tilde {\cal A}})$ and $[d^{2}h(\bbeta_{\tilde {\cal A}})]_{jk}=\bx'_j\bx_k+\rho\eta(\eta-1)|\tilde \beta_j|^{\eta-2}1_{\{j=k\}}$ for $j,k \in \tilde {\cal A}$. Then
\begin{eqnarray*}
\ln p(\by|\tilde {\cal A},\sigma^2) &\approx& q \ln \frac{\sqrt{\pi}\eta\rho^{1/\eta}}{\sqrt{2}\Gamma(1/\eta)} - \left( \frac{n-q}{2} + \frac{q}{\eta} \right) \ln \sigma^2 - \frac{h(\tilde \bbeta_{\tilde {\cal A}}) }{\sigma^2} - \frac 12 \log \det d^{2}h(\tilde{\bbeta}_{\tilde {\cal A}}) - \frac n2 \ln 2\pi.
\end{eqnarray*}
Plugging in the maximizing $\sigma^2$
\begin{equation*}
\tilde{\sigma}^2 = \frac{h(\tilde \bbeta_{\tilde {\cal A}})}{(n-q)/2+q/\eta}
\end{equation*}
and omitting the constant term $(n\ln 2 \pi)/2$, we obtain
\begin{equation*}
\mbox{EB}(\rho,\eta) \equiv - q \ln \frac{\sqrt{\pi}\eta\rho^{1/\eta}}{\sqrt{2}\Gamma(1/\eta)} + \left( \frac{n-q}{2} + \frac{q}{\eta} \right) \left( 1 + \ln \frac{h(\tilde \bbeta_{\tilde {\cal A}})}{(n-q)/2+q/\eta} \right) + \frac 12 \log \det d^{2}h(\tilde{\bbeta}_{\tilde {\cal A}}).
\end{equation*}
\end{proof}

\section{Further Applications}
\label{sec:gen-reg}

The generality of (\ref{eqn:pen-obj}) invites numerous applications beyond variable selection. After reparameterization, many generalized regularization problems are subject to the path following and empirical Bayes model selection procedure developed in the previous two sections. In this section we briefly discuss some further applications.

The recent articles \citep{TibshiraniTaylor10GenLasso,ZhouLange11LSPath,ZhouWu11EPSODE} consider the generalized $\ell_1$ regularization problem
\begin{align*}
    \min_{\bbeta} f(\bbeta) + \rho \|\bV \bbeta\|_1 + \rho \|\bW \bbeta\|_+,
\end{align*}
where $\|\ba\|_+ = \sum_i \max \{a_i,0\}$ is the sum of positive parts of its components. The first regularization term enforces equality constraints among coefficients at large $\rho$ while the second enforces inequality constraints. Applications range from $\ell_1$ penalized GLMs, shape restricted regressions, to nonparametric density estimation. For more parsimonious and unbiased solutions, generalized sparse regularization can be proposed
\begin{align}
    \min_{\bbeta} f(\bbeta) + \sum_{i=1}^r P(|\bv_i^t \bbeta|,\rho) + \sum_{j=1}^s P_+(\bw_j^t \bbeta,\rho),   \label{eqn:general-sparse}
\end{align}
where $P$ is a non-convex penalty function (power, double Pareto, SCAD, MC+, etc.) and $P_+(t,\rho)=P(t,\rho)$ for $t \ge 0$ and $P(0,\rho)$ otherwise. Devising an efficient path algorithm for (\ref{eqn:general-sparse}) is hard in general. However, when $\{\bv_i\}$ and $\{\bw_j\}$ are linearly independent, it can be readily solved by our path algorithm via a simple reparameterization. For ease of presentation, we only consider equality regularization here. Let the matrix $\bV \in \mathbb{R}^{r \times p}$ collect $\bv_i$ in its rows. The assumption of full row rank of $\bV$ implies $r \le p$. The trick is to reparameterize $\bbeta$ by $\bgamma = \tilde \bV \bbeta$ where $\tilde \bV \in \mathbb{R}^{p \times p}$ is the matrix $\bV$ appended with extra rows such that $\tilde \bV$ has full column rank. Then original coefficients $\bbeta$ can be recovered from the reparameterized ones $\bgamma$ via $\bbeta = (\bV^t\bV)^{-1} \bV^t \bgamma$. The reparameterized regularization problem is given by
\begin{align}
    \min_{\bgamma} f[(\bV^t\bV)^{-1} \bV^t \bgamma] + \sum_{j=1}^r P(|\gamma_j|,\rho) \label{eqn:general-sparse-transformed}
\end{align}
which is amenable to Algorithm \ref{algo:path}. Note that $f$ remains convex and twice differentiable under affine transformation of variables.

Regularization matrix $\bV$ with full row rank appears in numerous applications. For fused lasso, the regularization matrix
\begin{align*}
    \bV_1 = \left( \begin{array}{rrrrr}
    -1 & 1 &  \\
     & \ddots & \ddots  \\
     & & -1 & 1
    \end{array}  \right)
\end{align*}
has full row rank. In polynomial trend filtering \citep{KimKohBoyd09TrendFiltering,TibshiraniTaylor10GenLasso}, order $d$ finite differences between successive regression coefficients are penalized. Fused lasso corresponds to $d=1$ and the general polynomial trend filtering invokes regularization matrix $\bV_d = \bV_{d-1} \bV_1$, which again has full row rank. In Section \ref{sec:MandA}, cubic trend filtering for logistic regression is demonstrated on a financial data set. In all of these applications the regularization matrix $\bV$ is highly sparse and structured. The back transformation $(\bV^t\bV)^{-1} \bV^t \bgamma$ in (\ref{eqn:general-sparse-transformed}) is cheap to compute using a pre-computed sparse Cholesky factor of $\bV^t\bV$. The design matrix in terms of variable $\bgamma$ is $\bX(\bV^t\bV)^{-1} \bV^t$. In contrast to the regular variable selection problem, it shall not be assumed to be centered and scaled.

\section{Examples}
\label{sec:examples}

Various numerical examples in this section illustrate the path seeking algorithm and empirical Bayes model selection procedure developed in this article. The first two classical data sets show the mechanics of the path following for linear and logistic regressions and compare the model fit and prediction performance under various penalties. The third example illustrates the application of path algorithm and empirical Bayes procedure to cubic trend filtering in logistic regression using a financial data set. The last simulation example evaluates the computational efficiency of the path algorithm in a large $p$ small $n$ setting. Run times are displayed whenever possible to indicate the efficiency of our path following algorithm. The algorithm is run on a laptop with Intel Core i7 M620 2.66GHz CPU and 8 GB RAM. For reproducibility the code for all examples is available on the first author's web site.

\subsection{Linear regression: Prostate cancer data}

The first example concerns the classical prostate cancer data in \citep{Hastie09ESLBook}. The response variable is logarithm of prostate specific antigen ({\tt lpsa}) and the seven predictors are the logarithm of cancer volume ({\tt lweight}), {\tt age}, the logarithm of the amount of benign prostatic hyperplasia ({\tt lbph}), seminal vesicle invasion ({\tt svi}), the logarithm of capsular penetration ({\tt lcp}), Gleason score ({\tt gleason}), and percent of Gleason scores 4 or 5 ({\tt pgg45}). The data set contains 97 observations and is split into a training set of size 67 and a test set of 30 observations. 

Figure \ref{fig:solpath-compare} displays the solution paths of linear regression with nine representative penalties on the training set. Discontinuities occur in the paths from power family with $\eta=0.5$, continuous log penalty, and the log penalty with $\eta=1$. In contrast, the lasso solution path, from either enet(1) or power(1), is continuous and piecewise linear. Figure \ref{fig:solpath-compare} also illustrates the trade-off between continuity and unbiasedness. Using convex penalties, such as enet $(\eta=1,1.5,2)$ and power $(\eta=1)$, guarantees the continuity of solution path but causes bias in the estimates along the solution path. For a non-convex penalty such as power $(\eta=0.5)$, estimates are approximately unbiased once selected. However this can only be achieved by allowing discontinuities along the path.
\begin{figure}[h!]
\begin{center}
\includegraphics[width=6.5in]{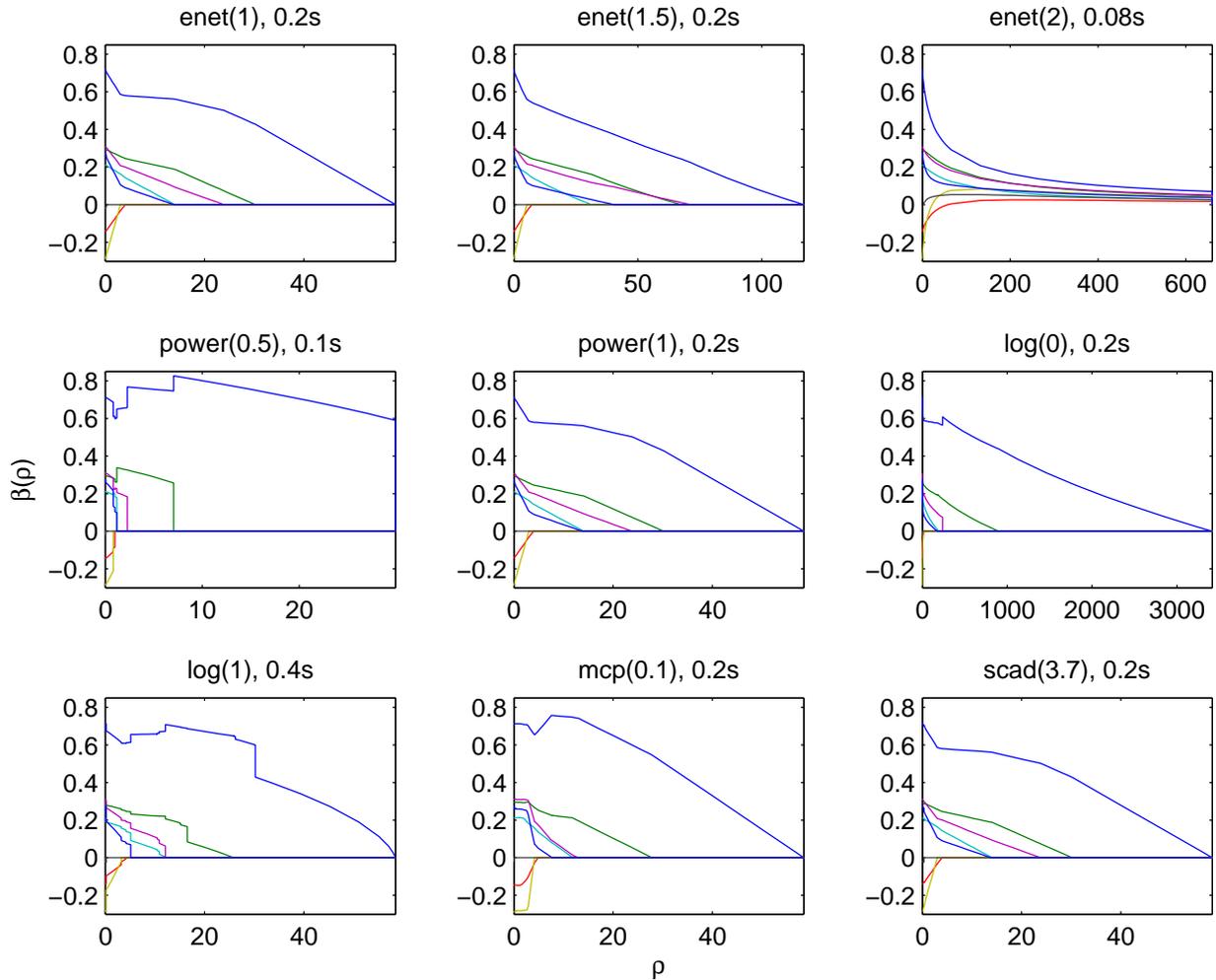}
\end{center}
\caption{Solution paths for the prostate cancer data.}
\label{fig:solpath-compare}
\end{figure}

To compare the model fit along the paths, it is more informative to plot the explained variation versus model dimension along the solution paths \citep{Friedman08GPS}. Upper panels of Figure \ref{fig:prostate-R2} display such plots for the enet, power, and log penalties at various penalty parameter values $\eta$. $y$-axis is the proportion $R^2(\rho)/R^2(0)$, i.e., the $R^2$ from the path solutions scaled by the maximum explained variation $R^2(0)$. Results for other penalties (MC+, SCAD) are not shown for brevity. Non-convex penalties show clear advantage in terms of higher explanatory power using fewer predictors. The model fit of path solutions in the test set shows similar patterns to those in Figure \ref{fig:DP-illustration}. To avoid repetition, they are not displayed here.

To evaluate the prediction performance, the prediction mean squared errors (MSE) on the test set from the solution paths are shown in the lower panels of Figure \ref{fig:prostate-R2}. Different classes of penalties all achieve the best prediction error of 0.45 with 4-6 predictors. It is interesting to note the highly concave penalties such as power ($\eta=0.2$) do not achieve the best prediction error along the path. Lasso and moderately concave penalties are quite competitive in achieving the best prediction error along the path. Convex penalties like enet with $\eta>1$ tend to admit too many predictors without achieving the best error rate.

\begin{figure}[h!]
\begin{center}
$$
\begin{array}{cc}
\includegraphics[width=2.5in]{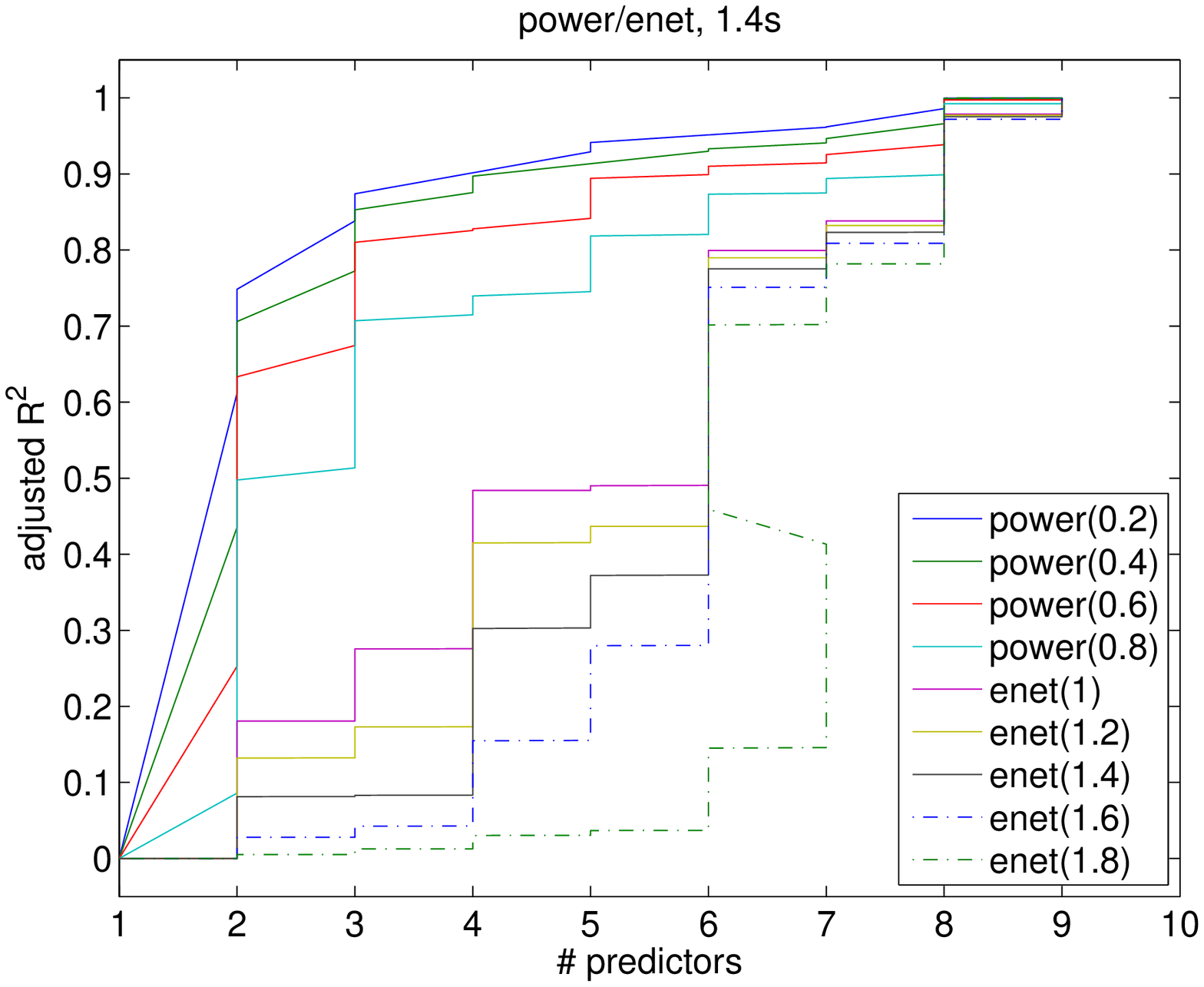} & \includegraphics[width=2.5in]{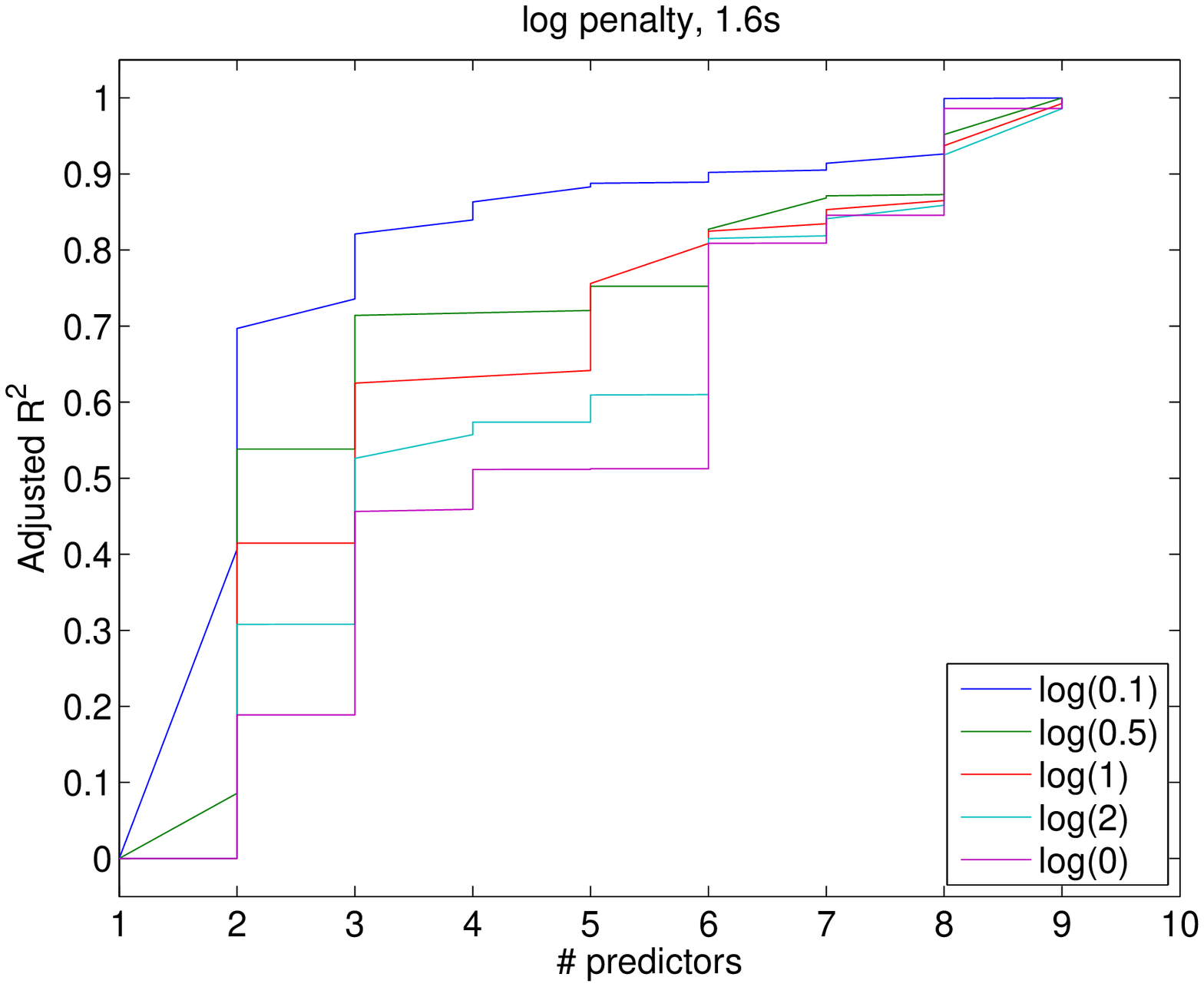} \\
\includegraphics[width=2.5in]{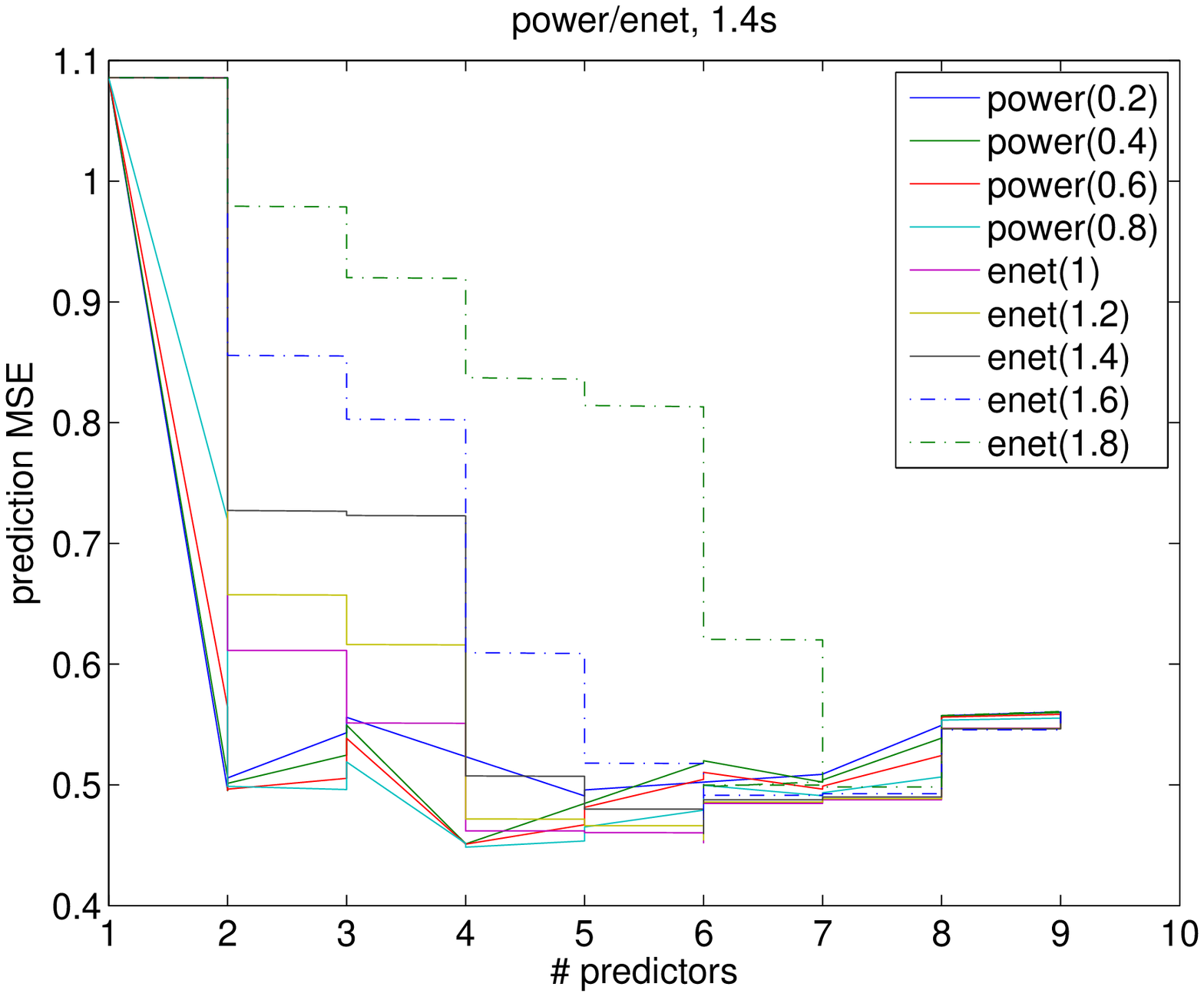} & \includegraphics[width=2.5in]{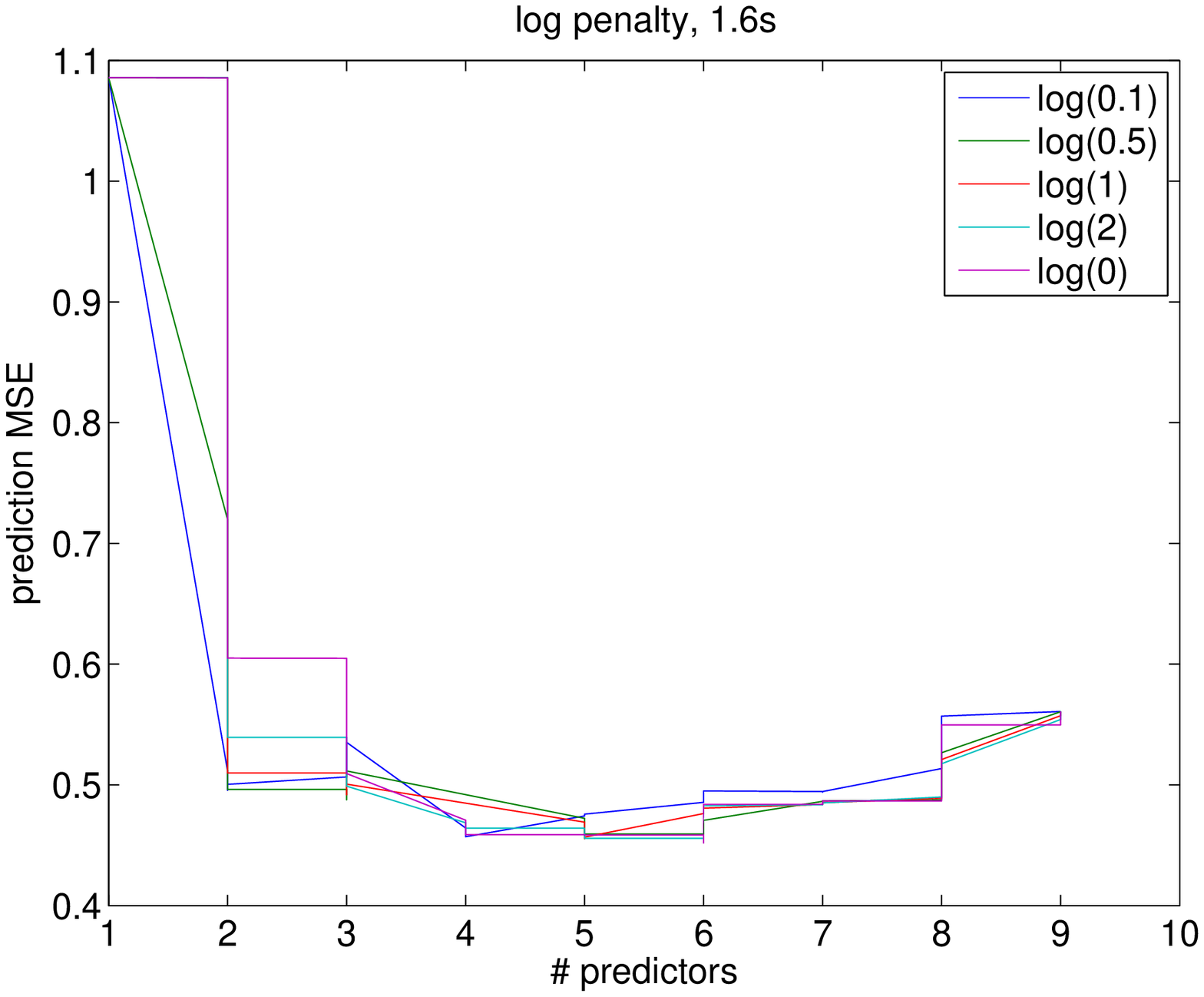}
\end{array}
$$
\end{center}
\caption{Upper panels: $R^2$ vs model dimension from various penalties for the prostate cancer data. Lower panels: Prediction mean square error (MSE) vs model dimension from various penalties for the prostate cancer data.}
\label{fig:prostate-R2}
\end{figure}

\subsection{Logistic regression: South Africa heart disease data}

For demonstration of logistic regression, we again use the classical South Africa heart disease data set in \citep{Hastie09ESLBook}. This data set has $n=462$ observations measured on 7 predictors. The response variable is binary (heart disease or not). We split the data set into a training set with 312 data points and a test set with 150 data points. Solution paths are obtained for the training data set from various penalties and are displayed in Figure \ref{fig:saheart-solpath}. Similar patterns are observed as those for the prostate cancer linear regression example. The discontinuities for concave penalties such as power ($\eta=0.5$) and log penalty ($\eta=1$) lead to less biased estimates along the paths. The plots of explained deviance versus model size for selected penalties are given in the upper panels of Figure \ref{fig:saheart-dev}. Solutions from concave penalties tend to explain more deviance with fewer predictors than lasso and enet with $\eta>1$. Deviance plots for the test set show a similar pattern. Prediction power of the path solutions is evaluated on the test data set and the prediction MSEs are reported in the lower panels of Figure \ref{fig:saheart-dev}. The highly concave penalties such as power $(\eta<1)$ and log penalty $(\eta=0.1)$ are able to achieve the best prediction error rate 0.425 with 5 predictors. Convex penalties and less concave ones perform worse in prediction power, even with more than 5 predictors.

\begin{figure}[h!]
\begin{center}
\includegraphics[width=6.5in]{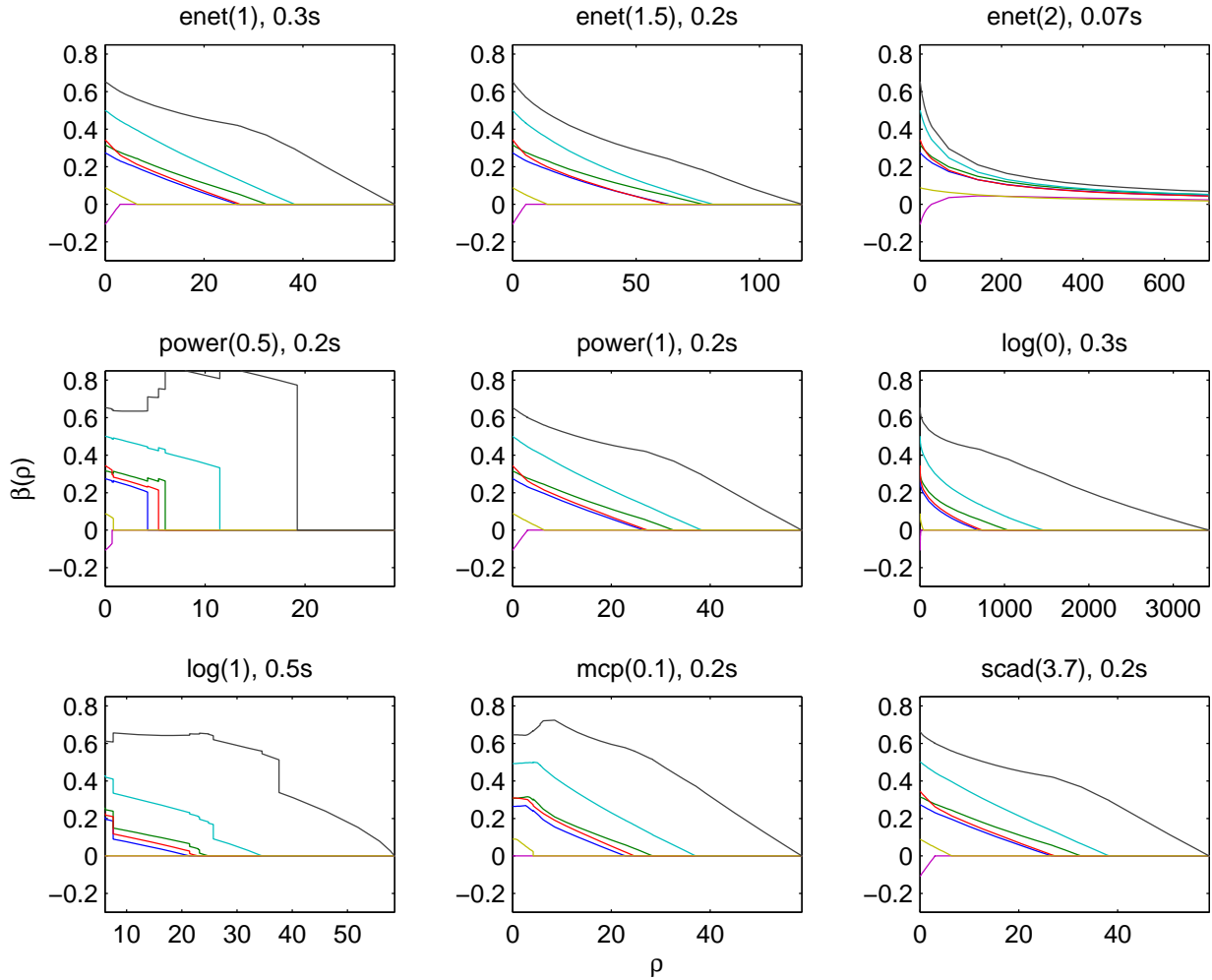}
\end{center}
\caption{Solution paths for the South Africa heart disease data.}
\label{fig:saheart-solpath}
\end{figure}

\begin{figure}[h!]
\begin{center}
$$
\begin{array}{cc}
\includegraphics[width=2.5in]{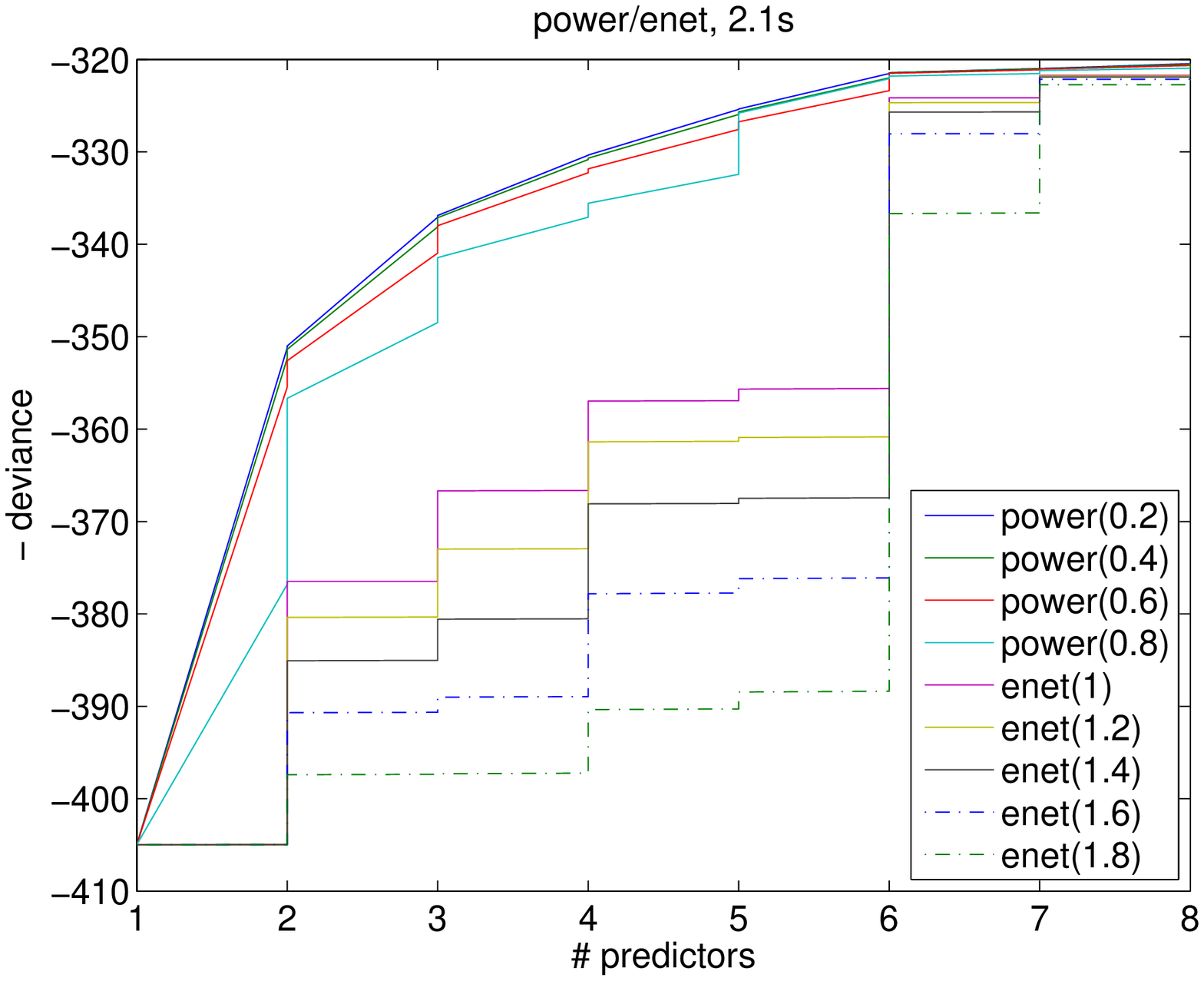} & \includegraphics[width=2.5in]{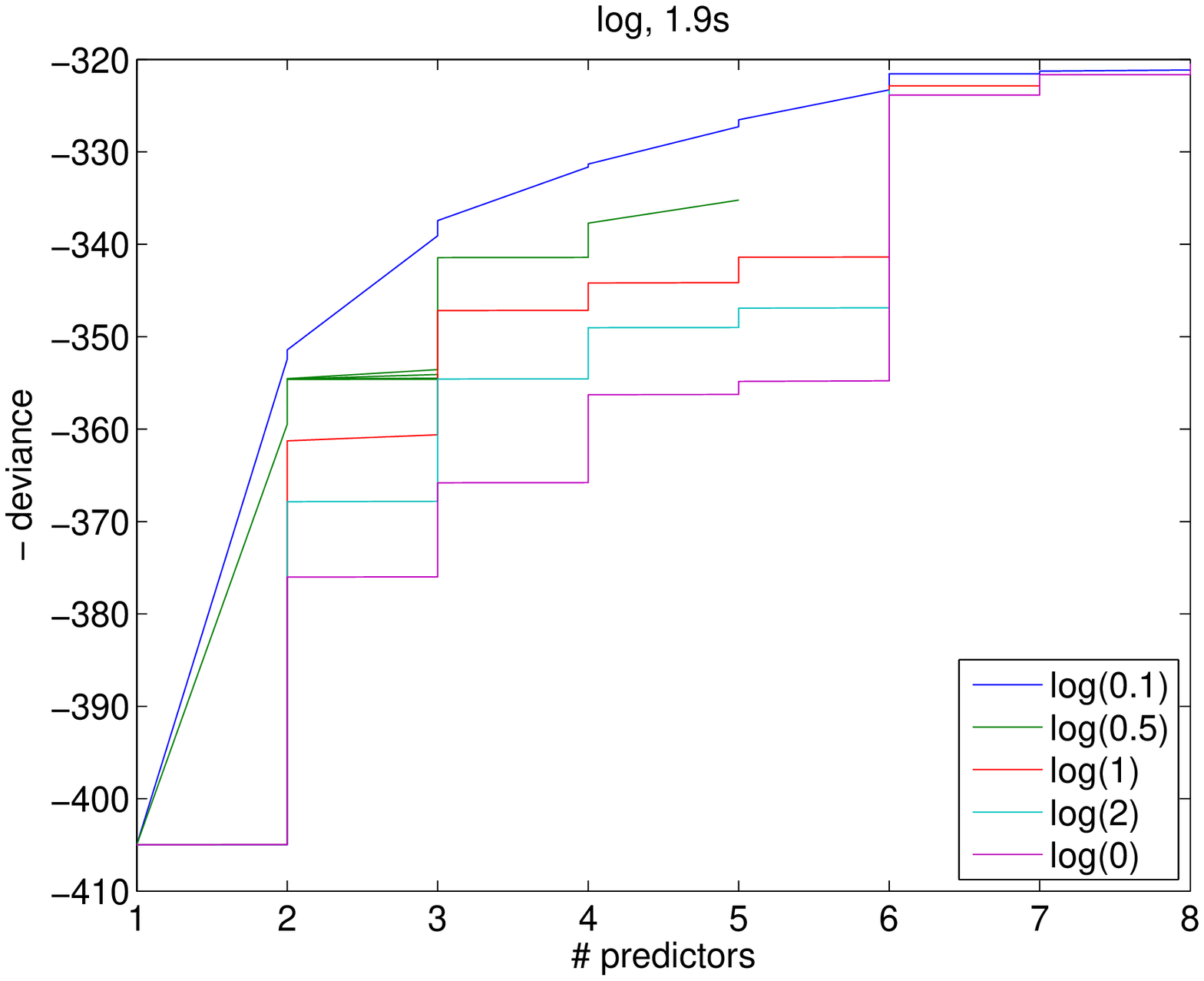} \\
\includegraphics[width=2.5in]{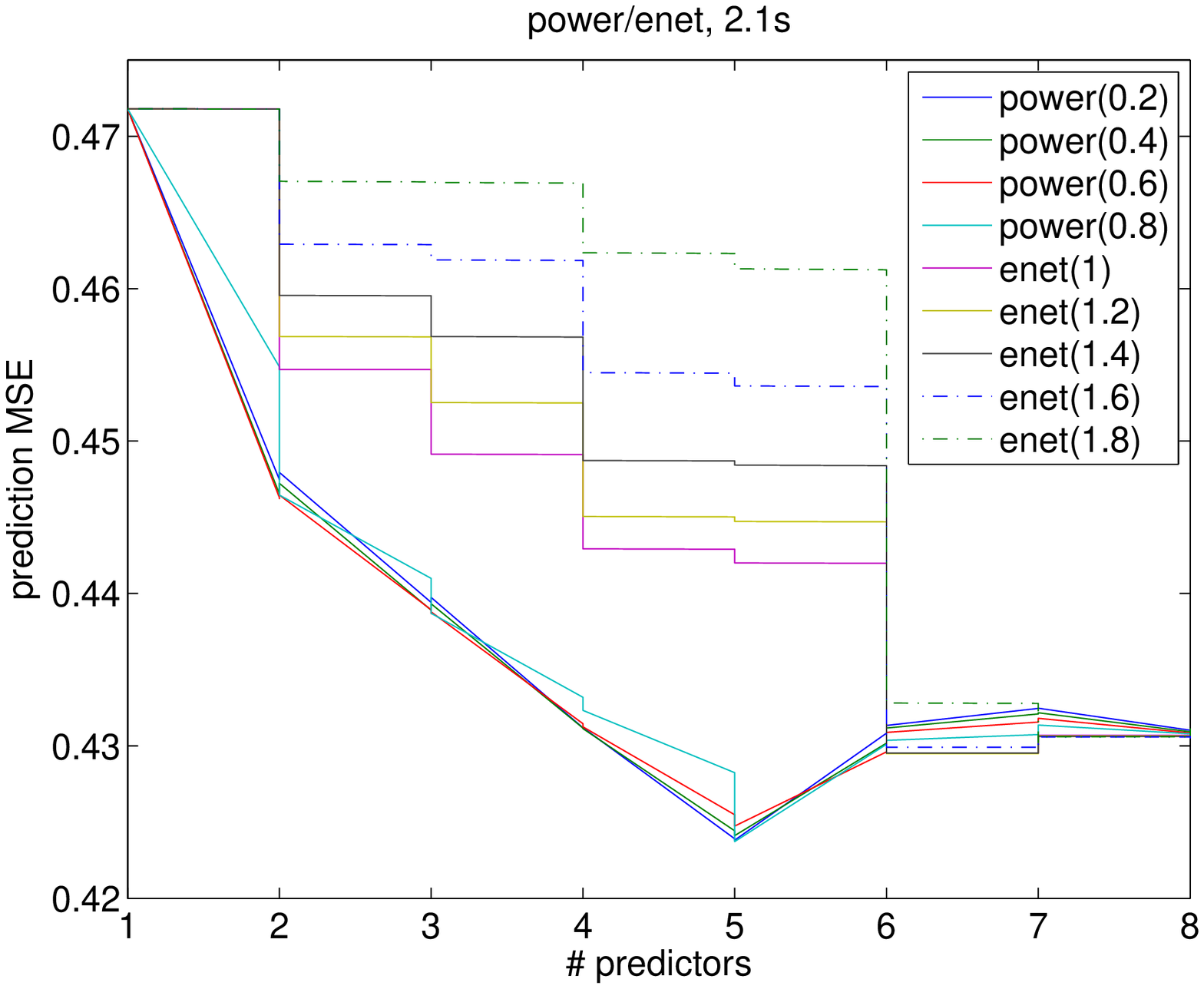} & \includegraphics[width=2.5in]{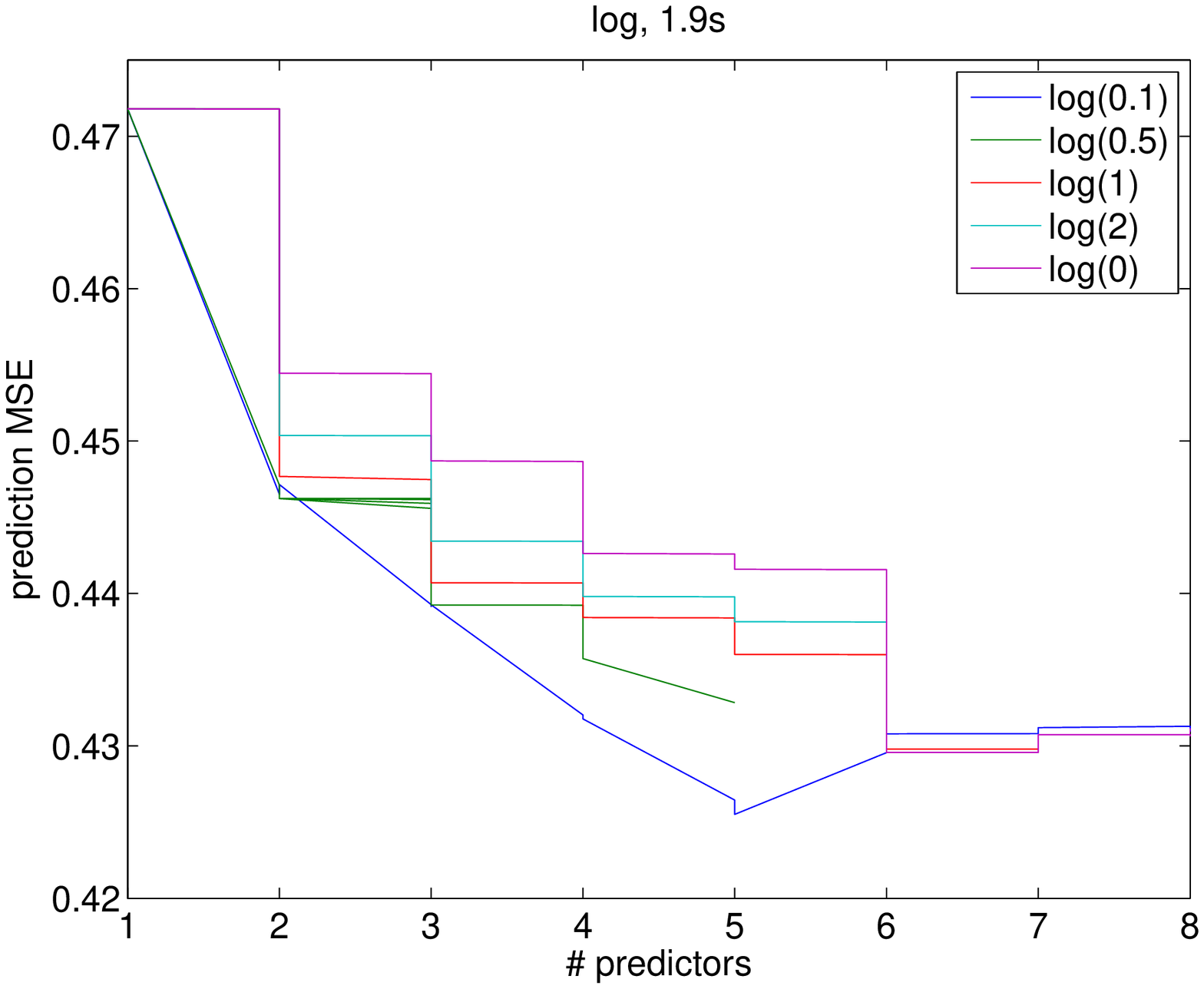}
\end{array}
$$
\end{center}
\caption{Upper panels: Negative deviance vs model dimension from various penalties for the South Africa heart disease data. Lower panels: Prediction mean square error (MSE) vs model dimension from various penalties for the South Africa heart disease data.}
\label{fig:saheart-dev}
\end{figure}

\subsection{Logistic regression with cubic trend filtering: M\&A data}
\label{sec:MandA}

The third example illustrates generalized regularization with logistic regression. We consider a merger and acquisition (M\&A) data set studied in recent articles \citep{ZhouWu11EPSODE,FanMaityWangWu11ManA}. This data set constitutes $n=1,371$ US companies with a binary response variable indicating whether the company becomes a leveraged buyout (LBO) target ($y_i=1$) or not ($y_i=0$). Seven covariates are recorded for each company. Table \ref{table:MandA-info} lists the 7 predictors and their p-values in the classical linear logistic regression. Predictors `long term investment', `log market equity', and `return S\&P 500 index' show no significance while the finance theory indicates otherwise.
\begin{table}[h!]
\begin{center}
\begin{tabular}{ll}
\toprule
Predictor & p-value  \\
\midrule
Cash Flow & 0.0019  \\
Case & 0.0211   \\
Long Term Investment & 0.5593 \\
Market to Book Ratio & 0.0000 \\
Log Market Equity & 0.5099  \\
Tax & 0.0358 \\
Return S\&P 500 Index & 0.2514  \\
\bottomrule
\end{tabular}
\end{center}
\caption{Predictors and their p-values from the linear logistic regression.}
\label{table:MandA-info}
\end{table}

To explore the possibly nonlinear effects of these quantitative covariates, we discretize each predictor into, say, 10 bins and fit a logistic regression. The first bin of each predictor is used as the reference level and effect coding is applied to each discretized covariate. The circles (o) in Figure \ref{fig:MandA-estimates} denote the estimated coefficients for each bin of each predictor and hint at some interesting nonlinear effects. For instance, the chance of being an LBO target seems to monotonically decease with market-to-book ratio and be quadratic as a function of log market equity. Regularization can be utilized to borrow strength between neighboring bins. The recent paper \citep{ZhouWu11EPSODE} applies cubic trend filtering to the 7 covariates using $\ell_1$ regularization. Here we demonstrate a similar regularization using a non-convex penalty. Specifically, we minimize a regularized negative logistic log-likelihood of the form
\begin{align*}
    - l(\bbeta_1,\ldots,\bbeta_7) + \sum_{j=1}^7 P_\eta(\bV_j \bbeta_j,\rho),
\end{align*}
where $\bbeta_j$ is the vector of regression coefficients for the $j$-th discretized covariate. The matrices in the regularization terms are specified as
\begin{align*}
    \bV_j &= \left( \begin{array}{rrrrrrrr}
    -1 & 2 & -1 \\
    1 & -4 & 6 & -4 & 1 \\
      & 1 & -4 & 6 & -4 & 1 &   \\
      & & & \ddots & \ddots & \ddots  \\
    &  & 1 & -4 & 6 & -4 & 1 &  \\
      & & & & & -1 & 2 & -1
    \end{array}  \right),
\end{align*}
which penalizes the fourth order finite differences between the bin estimates. Thus, as $\rho$ increases, the coefficient vectors for each covariate tend to be piecewise cubic with two ends being linear, mimicking the natural cubic spline. This is one example of polynomial trend filtering \citep{KimKohBoyd09TrendFiltering,TibshiraniTaylor10GenLasso} applied to logistic regression. Similar to semi-parametric regressions, regularizations in polynomial trend filtering `let the data speak for themselves'. In contrast, the bandwidth selection in semi-parametric regression is replaced by parameter tuning in regularizations. The number and locations of knots are automatically determined by the tuning parameter which is chosen according to a model selection criterion. The left panel of Figure \ref{fig:MandA-path} displayed the solution path delivered by the power penalty with $\eta=0.5$. It bridges the unconstrained estimates (denoted by o) to the constrained estimates (denoted by +). The right panel of Figure \ref{fig:MandA-path} shows the empirical Bayes criterion along the path. The dotted line in Figure \ref{fig:MandA-estimates} is the solution with smallest empirical Bayes criterion. It mostly matches the fully regularized solution except a small `dip' in the middle range of `log market equity'. The classical linear logistic regression corresponds to the restricted model where all bins for a covariate coincide. A formal analysis of deviance indicates that the regularized model at $\rho=2.5629$ is significant with respect to this null model with p-value 0.0023.

The quadratic or cubic like pattern in the effects of predictors `long term investment', `log market equity', and `return S\&P 500 index'  revealed by the regularized estimates explain why they are missed by the classical linear logistic regression. These patterns match some existing finance theories. For instance, Log of market equity is a measure of company size. Smaller companies are unpredictable in their profitability and extremely large companies are unlikely to be an LBO target because LBOs are typically financed with a large proportion of external debt. A company with a low cash flow is unlikely to be an LBO target because low cash flow is hard to meet the heavy debt burden associated with the LBO. On the other hand, a company carrying a high cash flow is likely to possess a new technology. It is risky to acquire such firms because it is hard to predict their profitability. The tax reason is obvious from the regularized estimates. The more tax the company is paying, the more tax benefits from an LBO.

\begin{figure}[h!]
\begin{center}
$$
\begin{array}{cc}
\includegraphics[width=3in]{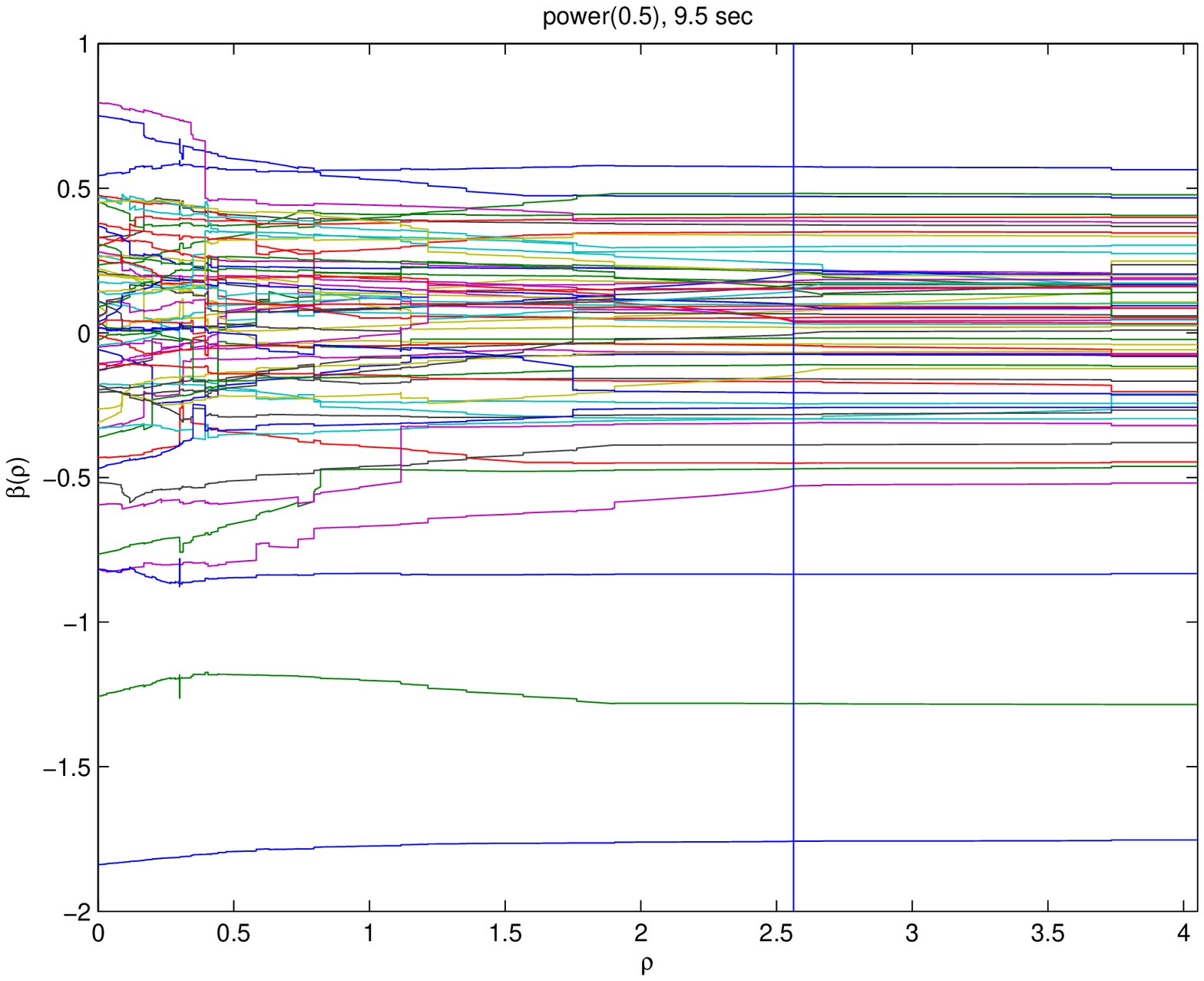} & \includegraphics[width=3in]{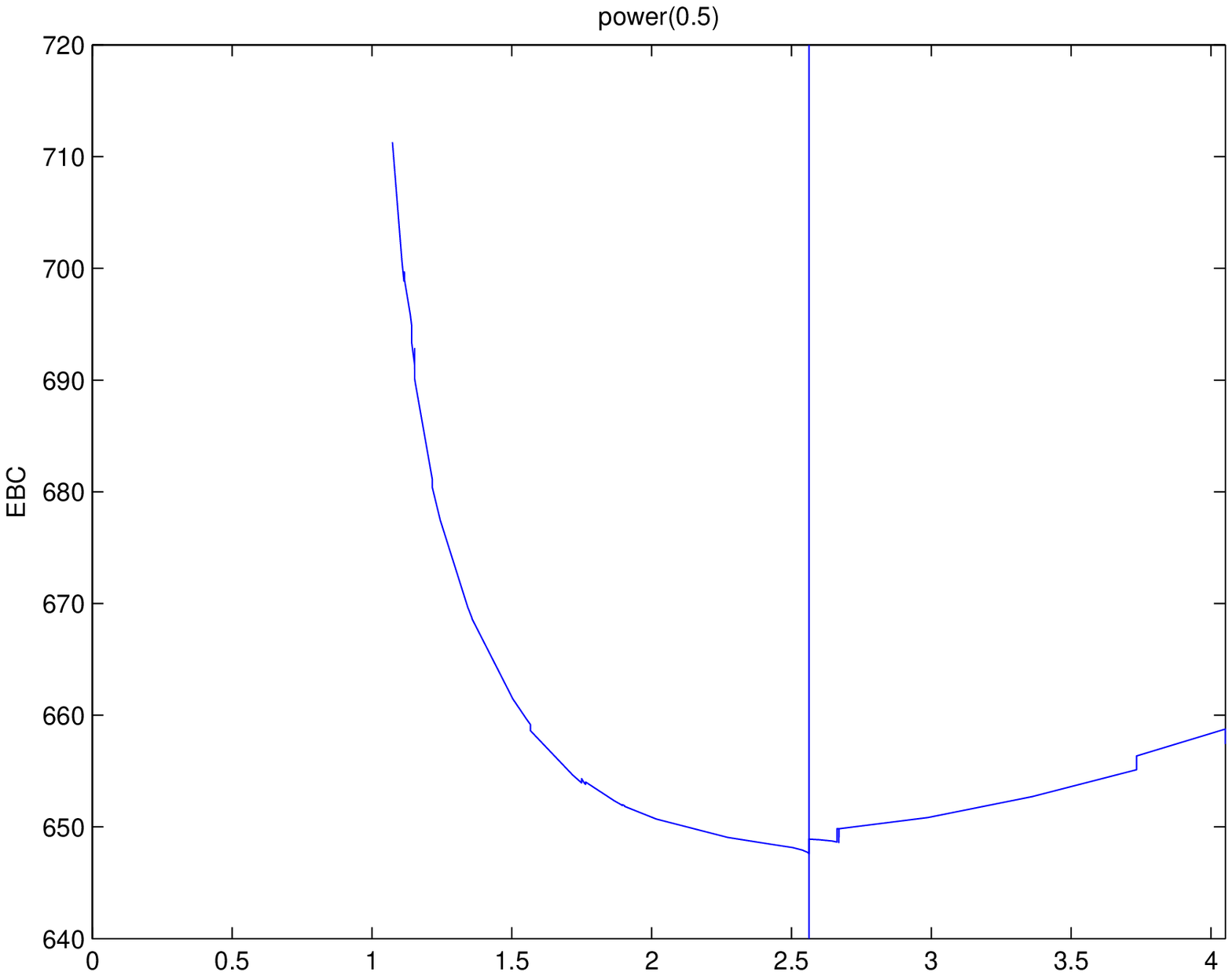}
\end{array}
$$
\caption{Regularized logistic regression on the M\&A data. Left: The trajectory of solution path from power penalty with $\eta=0.5$. Right: Empirical Bayes criterion along the path. Vertical lines indicate the model selected by the empirical Bayes criterion.}
\label{fig:MandA-path}
\end{center}
\end{figure}

\begin{figure}[h!]
\begin{center}
$$
\begin{array}{cc}
\includegraphics[width=6in]{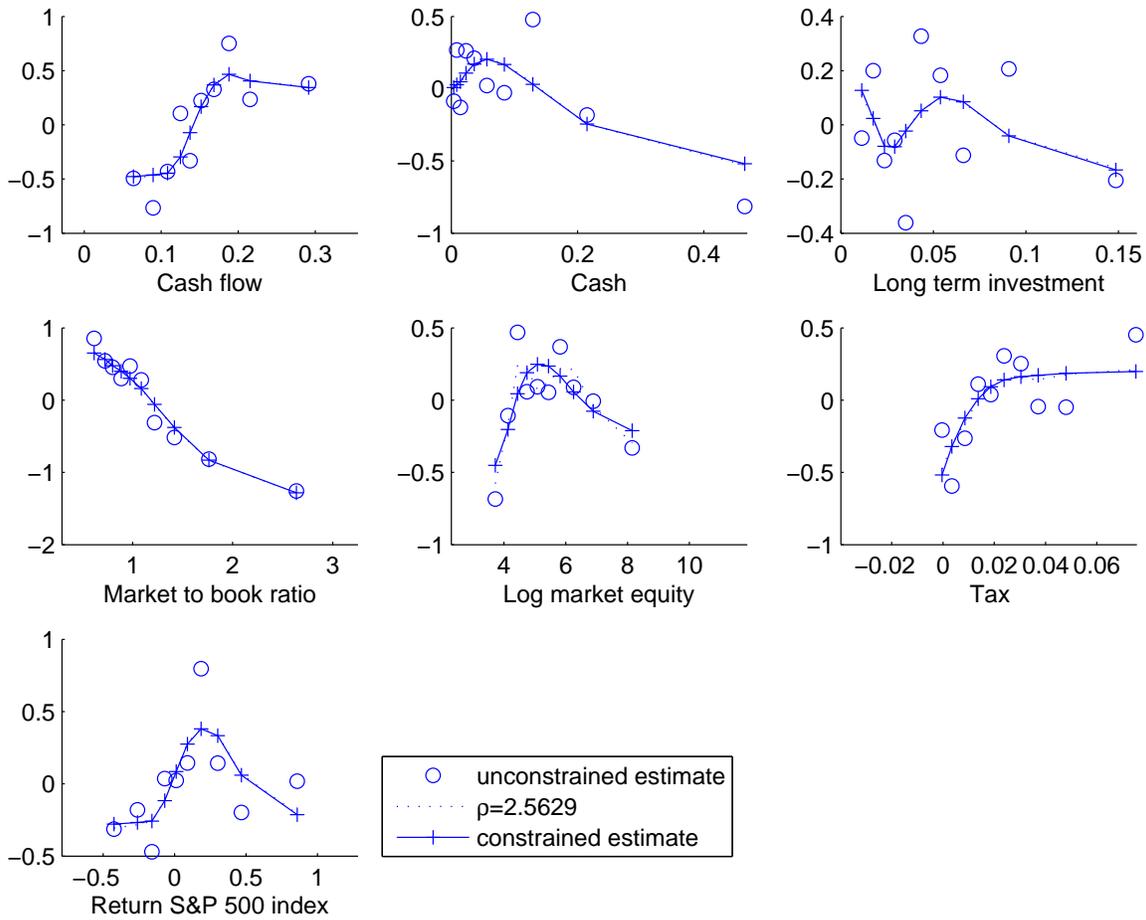}
\end{array}
$$
\caption{Snapshots of the path solution to the regularized logistic regression on the M\&A data set. The best model according to empirical Bayes criterion (dotted line) most matches the fully regularized solution (line with crosses) except that it has a dip in the middle part of the `log market equity' variable.}
\label{fig:MandA-estimates}
\end{center}
\end{figure}

\subsection{GLM sparse regression: large $p$, small $n$}
\label{sec:GLM-n200-p10000}

In all of the previous examples, the number of observations $n$ exceeds the number of parameters $p$. Our final simulation example evaluates the performance of the path following algorithm and empirical Bayes procedure in a large $p$ small $n$ setup for various generalized linear models (GLM). In each simulation replicate, $n=200$ independent responses $y_i$ are simulated from a normal (with unit variance), Poisson, and binomial distribution with mean $\mu_i$ respectively. The mean $\mu_i$ is determined by a $p=10,000$ dimensional covariate $\bx_i$ through a link function $g(\mu_i) = \alpha + \bx_i^t \bbeta$ where $\alpha$ is the intercept. Canonical links are used in the simulations. For the linear model, $g(\mu)=\mu$. For the Poisson model, $g(\mu) = \ln \mu$. For the logistic model, $g(\mu) = \log [\mu/(1-\mu)]$.

Numerous settings can be explored in this framework. For brevity and reproducibility, we only display the results for a simple exemplary setup: entries of covariate $\bx_i$ are generated from iid standard normal and the true regression coefficients are $\beta_i=3$ for $i=1,\ldots,5$, $\beta_i=-3$ for $i=6,\ldots,10$, and $\alpha=\beta_i=0$ for $i=11,\ldots,10,000$. Results presented in Figure \ref{fig:n200-p10000-timing}-\ref{fig:n100-p10000-mse} are based on 100 simulation replicates. In each replicate, path following is carried out under linear, Poisson, and logistic regression models coupled with power penalties at $\eta=0.25, 0.5, 0.75, 1$, representing a spectrum from (nearly) best subset regression to lasso regression. Results for other penalties are not shown to save space. Path following is terminated when at least 100 predictors are selected or separation is detected in the Poisson or logistic models, whichever occurs first. Results for the logistic sparse regression have to be interpreted with caution due to frequent occurrence of complete separation along solution paths. This is common in large $p$ small $n$ problems as the chance of finding a linear combination of a few predictors that perfectly predicts the $n=200$ binary responses is very high when there are $p=10,000$ candidate covariates. Therefore the results for logistic regression largely reflect the quality of solutions when path following terminates at complete separation.

Figure \ref{fig:n200-p10000-timing} displays the boxplots of run times at different combinations of the GLM model and penalty value $\eta$. The majority of runs take less than one minute across all models, except Poisson regression with $\eta=0.75$. The run times in this setting display large variability with a median around 50 seconds. Logistic regression with non-convex penalties ($\eta=0.25, 0.5, 0,75$) has shorter run times than lasso penalty ($\eta=1$) due to complete separation at early stages of path following.

Figures \ref{fig:n100-p10000-fpr} and \ref{fig:n100-p10000-fnr} display the false positive rate (FPR) and false negative rate (FNR) of the model selected by the empirical Bayes procedure at different combinations of GLM model and penalty value $\eta$. FPR (type I error rate) is defined as the proportion of false positives in the selected model among all true negatives (9990 in this case). FNR (type II error rate) is defined as the proportion of false negatives in the selected model among all true positives (10 in this case). These two numbers give a rough measure of model selection performance. For all three GLM models, power penalties with larger $\eta$ (close to convexity) tend to select more predictors, leading to significantly higher FPR. On the other hand, the median FNR appears not significantly improved in larger $\eta$ cases, although they admit more predictors. This indicates the overall improved model selection performance of non-convex penalties.

More interesting is the mean square error (MSE) of the parameter estimate $\tilde \bbeta$ at the model selected by the empirical Bayes procedure. MSE is defined as [$\sum_j (\tilde \beta_j - \beta)^2 / p]^{1/2}$. Figure \ref{fig:n100-p10000-mse} shows that lasso ($\eta=1$) has risk properties comparable to the non-convex penalties, although it is a poor model selector in terms of FPR and FNR.

We should keep in mind that these results are particular to the specific simulation setting we presented here and may vary across numerous factors such as pairwise correlations between the covariates, signal to noise ratio, sample size $n$ and dimension $p$, penalty type, etc. We hope that the tools developed in this article facilitate such comparative studies. The generality of our method precludes extensive numerical comparison with current methods as only a few software packages are available for the special cases of \eqref{eqn:pen-obj}. In supplementary materials, we compare the run times of our algorithm to that of \citet{FriedmanHastieTibshirani10GLMCD} for the special case of lasso penalized GLM.

\begin{figure}[!ht]
\begin{center}
$$
\begin{array}{c}
\includegraphics[width=4.5in]{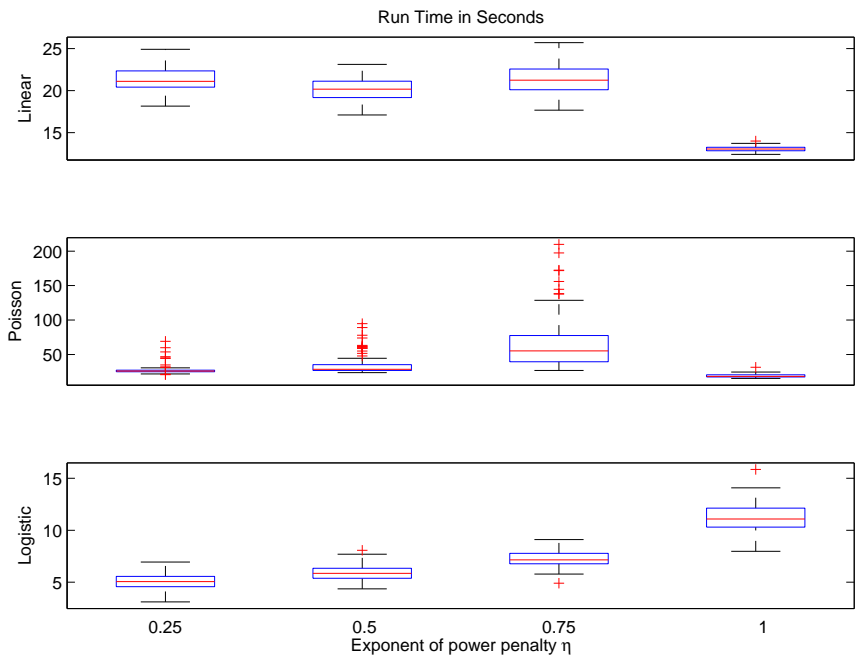}
\end{array}
$$
\caption{Run times of path following for GLM sparse regression with power penalties from 100 replicates. Problem size is $n=200$ and $p=10,000$.}
\label{fig:n200-p10000-timing}
\end{center}
\end{figure}

\begin{figure}[h!]
\begin{center}
$$
\begin{array}{c}
\includegraphics[width=4in]{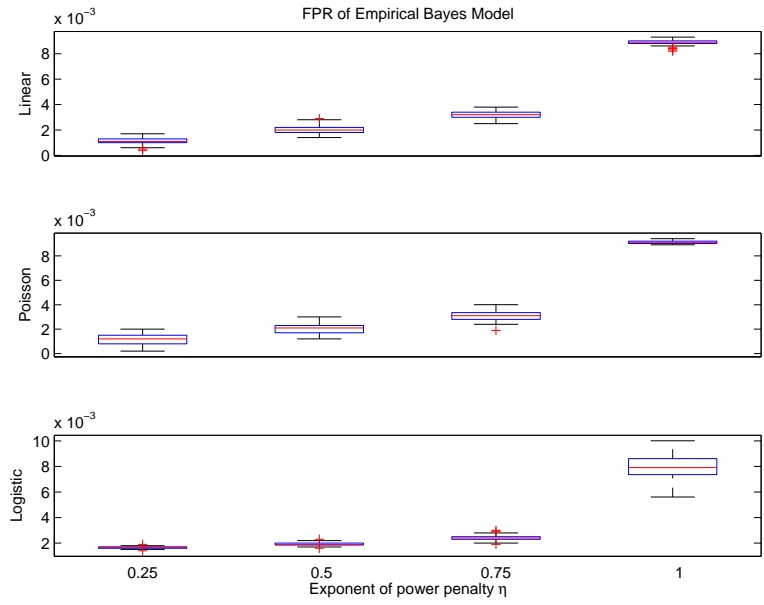}
\end{array}
$$
\caption{False positive rate (FPR) of the GLM sparse model selected by the empirical Bayes criterion.}
\label{fig:n100-p10000-fpr}
\end{center}
\end{figure}

\begin{figure}[h!]
\begin{center}
$$
\begin{array}{c}
\includegraphics[width=4in]{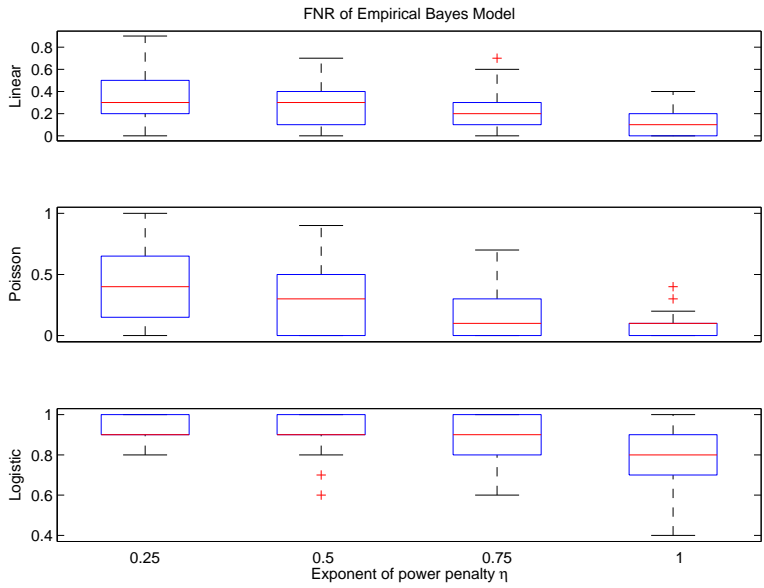}
\end{array}
$$
\caption{False negative rate (FNR) of the GLM sparse model selected by the empirical Bayes criterion.}
\label{fig:n100-p10000-fnr}
\end{center}
\end{figure}

\begin{figure}[h!]
\begin{center}
$$
\begin{array}{c}
\includegraphics[width=4in]{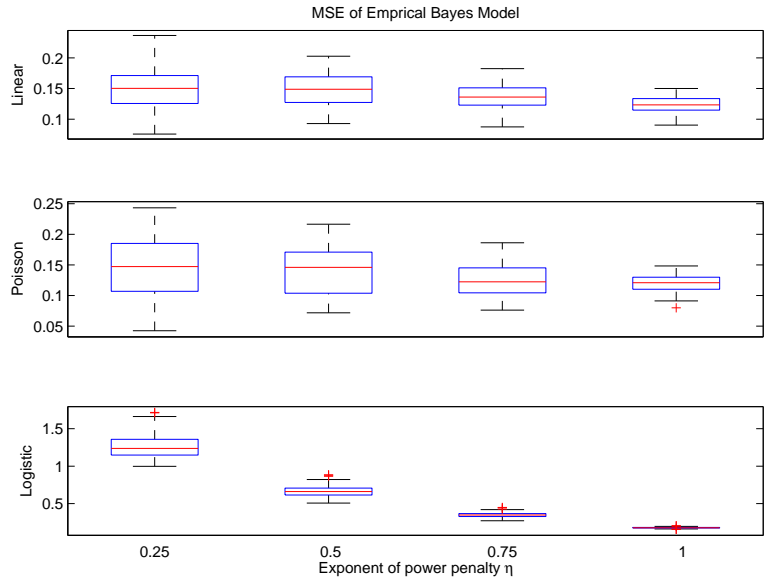}
\end{array}
$$
\caption{Mean square error (MSE) of the parameter estimate from the model selected by the empirical Bayes criterion.}
\label{fig:n100-p10000-mse}
\end{center}
\end{figure}

\section{Discussion}
\label{sec:discussion}

In this article we propose a generic path following algorithm for any combination of a convex loss function and a penalty function that satisfies mild conditions. Although motivated by the unpublished work by \citet{Friedman08GPS}, our algorithm turns out to be different from his general path seeking (GPS) algorithm. Further research is needed on the connection between the two. The ODE approach for path following tracks the solution smoothly and avoids the need to choose a fixed step size as required by most currently available regularization path algorithms.

Motivated by a shrinkage prior in the Bayesian setting, we derived an empirical Bayes procedure that allows quick search for a model and the corresponding tuning parameter along the solution paths from a large class of penalty functions. All necessary quantities for the empirical Bayes procedure naturally arise in the path algorithm.

Besides sparse regression, simple reparameterization extends the applicability of the path algorithm to many more generalized regularization problems. The cubic trend filtering example with the M\&A data illustrates the point.

Our numerical examples illustrate the working mechanics of the path algorithm and properties of different penalties. A more extensive comparative study of the penalties in various situations is well deserved. The tools developed in this article free statisticians from the often time consuming task of developing optimization algorithms for specific loss and penalty combination. Interested readers are welcome to use the {\tt SparseReg} toolbox freely available on the first author's web site.


\bibliographystyle{Chicago}
\bibliography{pareto-notes}

\begin{thebibliography}{}

\bibitem[\protect\citeauthoryear{Akaike}{Akaike}{1974}]{Akaike74AIC}
Akaike, H. (1974).
\newblock {A new look at the statistical model identification}.
\newblock {\em IEEE Transactions on Automatic Control\/}~{\em 19\/}(6),
  716--723.

\bibitem[\protect\citeauthoryear{Armagan}{Armagan}{2009}]{Armagan09Bridge}
Armagan, A. (2009).
\newblock Variational bridge regression.
\newblock {\em Journal of Machine Learning Research\/}~{\em 5}, 17--24.

\bibitem[\protect\citeauthoryear{Armagan, Dunson, and Lee}{Armagan
  et~al.}{2011}]{Armagan11Pareto}
Armagan, A., D.~Dunson, and J.~Lee (2011).
\newblock Generalized double {P}areto shrinkage.
\newblock {\em arxiv:1104.0861v1\/}.

\bibitem[\protect\citeauthoryear{Cand{\`e}s, Wakin, and Boyd}{Cand{\`e}s
  et~al.}{2008}]{Candes08ReWtL1}
Cand{\`e}s, E.~J., M.~B. Wakin, and S.~P. Boyd (2008).
\newblock Enhancing sparsity by reweighted {$l_1$} minimization.
\newblock {\em J. Fourier Anal. Appl.\/}~{\em 14\/}(5-6), 877--905.

\bibitem[\protect\citeauthoryear{Chen, Donoho, and Saunders}{Chen
  et~al.}{2001}]{ChenDonohoSaunders01BasisPursuit}
Chen, S.~S., D.~L. Donoho, and M.~A. Saunders (2001).
\newblock Atomic decomposition by basis pursuit.
\newblock {\em SIAM Rev.\/}~{\em 43\/}(1), 129--159.

\bibitem[\protect\citeauthoryear{Efron, Hastie, Johnstone, and
  Tibshirani}{Efron et~al.}{2004}]{EfronHastieIainTibshirani04LARS}
Efron, B., T.~Hastie, I.~Johnstone, and R.~Tibshirani (2004).
\newblock Least angle regression.
\newblock {\em Ann. Statist.\/}~{\em 32\/}(2), 407--499.
\newblock With discussion, and a rejoinder by the authors.

\bibitem[\protect\citeauthoryear{Fan and Li}{Fan and Li}{2001}]{FanLi01SCAD}
Fan, J. and R.~Li (2001).
\newblock Variable selection via nonconcave penalized likelihood and its oracle
  properties.
\newblock {\em J. Amer. Statist. Assoc.\/}~{\em 96\/}(456), 1348--1360.

\bibitem[\protect\citeauthoryear{Fan, Maity, Wang, and Wu}{Fan
  et~al.}{2011}]{FanMaityWangWu11ManA}
Fan, J., A.~Maity, Y.~Wang, and Y.~Wu (2011).
\newblock Analyzing mergers and acquisition data using parametrically guided
  generalized additive models.
\newblock ~{\em submitted}.

\bibitem[\protect\citeauthoryear{Frank and Friedman}{Frank and
  Friedman}{1993}]{FrankFriedman93Bridge}
Frank, I.~E. and J.~H. Friedman (1993).
\newblock A statistical view of some chemometrics regression tools.
\newblock {\em Technometrics\/}~{\em 35\/}(2), 109--135.

\bibitem[\protect\citeauthoryear{Friedman}{Friedman}{2008}]{Friedman08GPS}
Friedman, J. (2008).
\newblock Fast sparse regression and classification.
\newblock {\em http://www-stat.stanford.edu/~jhf/ftp/GPSpaper.pdf\/}.

\bibitem[\protect\citeauthoryear{Friedman, Hastie, H{\"o}fling, and
  Tibshirani}{Friedman et~al.}{2007}]{Friedman07Coordinate}
Friedman, J., T.~Hastie, H.~H{\"o}fling, and R.~Tibshirani (2007).
\newblock Pathwise coordinate optimization.
\newblock {\em Ann. Appl. Stat.\/}~{\em 1\/}(2), 302--332.

\bibitem[\protect\citeauthoryear{Friedman, Hastie, and Tibshirani}{Friedman
  et~al.}{2010}]{FriedmanHastieTibshirani10GLMCD}
Friedman, J.~H., T.~Hastie, and R.~Tibshirani (2010, 2).
\newblock Regularization paths for generalized linear models via coordinate
  descent.
\newblock {\em Journal of Statistical Software\/}~{\em 33\/}(1), 1--22.

\bibitem[\protect\citeauthoryear{Fu}{Fu}{1998}]{Fu98BridgeLasso}
Fu, W.~J. (1998).
\newblock Penalized regressions: the bridge versus the lasso.
\newblock {\em Journal of Computational and Graphical Statistics\/}~{\em
  7\/}(3), 397--416.

\bibitem[\protect\citeauthoryear{Hastie, Tibshirani, and Friedman}{Hastie
  et~al.}{2009}]{Hastie09ESLBook}
Hastie, T., R.~Tibshirani, and J.~Friedman (2009).
\newblock {\em The Elements of Statistical Learning: Data Mining, Inference,
  and Prediction\/} (Second ed.).
\newblock Springer Series in Statistics. New York: Springer.

\bibitem[\protect\citeauthoryear{Hoerl and Kennard}{Hoerl and
  Kennard}{1970}]{HoerlKennard70Ridge}
Hoerl, A.~E. and R.~W. Kennard (1970).
\newblock Ridge regression: biased estimation for nonorthogonal problems.
\newblock {\em Technometrics\/}~{\em 12}, 55--67.

\bibitem[\protect\citeauthoryear{Hunter and Li}{Hunter and
  Li}{2005}]{HunterLi05MM}
Hunter, D.~R. and R.~Li (2005).
\newblock Variable selection using {MM} algorithms.
\newblock {\em Ann. Statist.\/}~{\em 33\/}(4), 1617--1642.

\bibitem[\protect\citeauthoryear{Kim, Koh, Boyd, and Gorinevsky}{Kim
  et~al.}{2009}]{KimKohBoyd09TrendFiltering}
Kim, S.-J., K.~Koh, S.~Boyd, and D.~Gorinevsky (2009).
\newblock {$l_1$} trend filtering.
\newblock {\em SIAM Rev.\/}~{\em 51\/}(2), 339--360.

\bibitem[\protect\citeauthoryear{Knight and Fu}{Knight and
  Fu}{2000}]{KnightFu00LassoAsymptotics}
Knight, K. and W.~Fu (2000).
\newblock Asymptotics for lasso-type estimators.
\newblock {\em The Annals of Statistics\/}~{\em 28\/}(5), 1356--1378.

\bibitem[\protect\citeauthoryear{Lange}{Lange}{2004}]{Lange04Optm}
Lange, K. (2004).
\newblock {\em Optimization}.
\newblock Springer Texts in Statistics. New York: Springer-Verlag.

\bibitem[\protect\citeauthoryear{Lange}{Lange}{2010}]{Lange10NumAnalBook}
Lange, K. (2010).
\newblock {\em Numerical Analysis for Statisticians\/} (Second ed.).
\newblock Statistics and Computing. New York: Springer.

\bibitem[\protect\citeauthoryear{Mazumder, Friedman, and Hastie}{Mazumder
  et~al.}{2011}]{Mazumder11SparseNet}
Mazumder, R., J.~Friedman, and T.~Hastie (2011).
\newblock {SparseNet}: coordinate descent with non-convex penalties.
\newblock {\em JASA\/}.

\bibitem[\protect\citeauthoryear{Meinshausen and B\"uhlmann}{Meinshausen and
  B\"uhlmann}{2006}]{MeinshausenBuhlmann06GraphLasso}
Meinshausen, N. and P.~B\"uhlmann (2006).
\newblock {High-dimensional graphs and variable selection with the lasso}.
\newblock {\em Ann. Statist.\/}~{\em 34\/}(3), 1436--1462.

\bibitem[\protect\citeauthoryear{Osborne, Presnell, and Turlach}{Osborne
  et~al.}{2000}]{OsbornePresnellTurlach00LassoAlgo}
Osborne, M.~R., B.~Presnell, and B.~A. Turlach (2000).
\newblock A new approach to variable selection in least squares problems.
\newblock {\em IMA J. Numer. Anal.\/}~{\em 20\/}(3), 389--403.

\bibitem[\protect\citeauthoryear{Park and Hastie}{Park and
  Hastie}{2007}]{ParkHastie07GLMLasso}
Park, M.~Y. and T.~Hastie (2007).
\newblock {$L_1$}-regularization path algorithm for generalized linear models.
\newblock {\em J. R. Stat. Soc. Ser. B Stat. Methodol.\/}~{\em 69\/}(4),
  659--677.

\bibitem[\protect\citeauthoryear{Ruszczy{\'n}ski}{Ruszczy{\'n}ski}{2006}]{Rusz%
czynski06Book}
Ruszczy{\'n}ski, A. (2006).
\newblock {\em Nonlinear Optimization}.
\newblock Princeton, NJ: Princeton University Press.

\bibitem[\protect\citeauthoryear{Schwarz}{Schwarz}{1978}]{Schwartz78BIC}
Schwarz, G. (1978).
\newblock Estimating the dimension of a model.
\newblock {\em The Annals of Statistics\/}~{\em 6\/}(2), 461--464.

\bibitem[\protect\citeauthoryear{Soetaert, Petzoldt, and Setzer}{Soetaert
  et~al.}{2010}]{Soetaert10RdeSolve}
Soetaert, K., T.~Petzoldt, and R.~W. Setzer (2010).
\newblock Solving differential equations in {R}: Package {deSolve}.
\newblock {\em Journal of Statistical Software\/}~{\em 33\/}(9), 1--25.

\bibitem[\protect\citeauthoryear{Tibshirani}{Tibshirani}{1996}]{Tibshirani96La%
sso}
Tibshirani, R. (1996).
\newblock Regression shrinkage and selection via the lasso.
\newblock {\em J. Roy. Statist. Soc. Ser. B\/}~{\em 58\/}(1), 267--288.

\bibitem[\protect\citeauthoryear{Tibshirani and Taylor}{Tibshirani and
  Taylor}{2011}]{TibshiraniTaylor10GenLasso}
Tibshirani, R. and J.~Taylor (2011).
\newblock The solution path of the generalized lasso.
\newblock {\em Ann. Statist.\/}~{\em to appear}.

\bibitem[\protect\citeauthoryear{Wang and Leng}{Wang and
  Leng}{2007}]{WangLeng07Lasso}
Wang, H. and C.~Leng (2007).
\newblock Unified lasso estimation by least squares approximation.
\newblock {\em Journal of the American Statistical Association\/}~{\em 102},
  1039--1048.

\bibitem[\protect\citeauthoryear{Wang, Li, and Tsai}{Wang
  et~al.}{2007}]{Wang07SCADBIC}
Wang, H., R.~Li, and C.-L. Tsai (2007).
\newblock Tuning parameter selectors for the smoothly clipped absolute
  deviation method.
\newblock {\em Biometrika\/}~{\em 94\/}(3), 553--568.

\bibitem[\protect\citeauthoryear{Wu and Lange}{Wu and
  Lange}{2008}]{WuLange08Coordinate}
Wu, T.~T. and K.~Lange (2008).
\newblock Coordinate descent algorithms for lasso penalized regression.
\newblock {\em Ann. Appl. Stat.\/}~{\em 2\/}(1), 224--244.

\bibitem[\protect\citeauthoryear{Wu}{Wu}{2011}]{Wu10ODELasso}
Wu, Y. (2011).
\newblock An ordinary differential equation-based solution path algorithm.
\newblock {\em Journal of Nonparametric Statistics\/}~{\em 23}, 185--199.

\bibitem[\protect\citeauthoryear{Yuan and Lin}{Yuan and
  Lin}{2005}]{YuanLin05EmpiricalBayes}
Yuan, M. and Y.~Lin (2005).
\newblock Efficient empirical {B}ayes variable selection and estimation in
  linear models.
\newblock {\em Journal of the American Statistical Association\/}~{\em 100},
  1215--1225.

\bibitem[\protect\citeauthoryear{Zhang}{Zhang}{2010}]{Zhang10MCP}
Zhang, C.-H. (2010).
\newblock Nearly unbiased variable selection under minimax concave penalty.
\newblock {\em Ann. Statist.\/}~{\em 38\/}(2), 894--942.

\bibitem[\protect\citeauthoryear{Zhang and Huang}{Zhang and
  Huang}{2008}]{ZhangHuang08LassoBias}
Zhang, C.-H. and J.~Huang (2008).
\newblock The sparsity and bias of the lasso selection in high-dimensional
  linear regression.
\newblock {\em Annals of Statistics\/}~{\em 36}, 1567--1594.

\bibitem[\protect\citeauthoryear{Zhao and Yu}{Zhao and
  Yu}{2006}]{ZhaoYu06LassoConsistency}
Zhao, P. and B.~Yu (2006).
\newblock On model selection consistency of {L}asso.
\newblock {\em J. Mach. Learn. Res.\/}~{\em 7}, 2541--2563.

\bibitem[\protect\citeauthoryear{Zhou and Lange}{Zhou and
  Lange}{2011}]{ZhouLange11LSPath}
Zhou, H. and K.~Lange (2011).
\newblock A path algorithm for constrained estimation.
\newblock {\em arXiv:1103.3738\/}.

\bibitem[\protect\citeauthoryear{Zhou and Wu}{Zhou and
  Wu}{2011}]{ZhouWu11EPSODE}
Zhou, H. and Y.~Wu (2011).
\newblock A generic path algorithm for regularized statistical estimation.
\newblock {\em submitted\/}.

\bibitem[\protect\citeauthoryear{Zou}{Zou}{2006}]{Zou06AdaptiveLasso}
Zou, H. (2006).
\newblock The adaptive lasso and its oracle properties.
\newblock {\em Journal of the American Statistical Association\/}~{\em
  101\/}(476), 1418--1429.

\bibitem[\protect\citeauthoryear{Zou and Hastie}{Zou and
  Hastie}{2005}]{ZouHastie05Enet}
Zou, H. and T.~Hastie (2005).
\newblock Regularization and variable selection via the elastic net.
\newblock {\em J. R. Stat. Soc. Ser. B Stat. Methodol.\/}~{\em 67\/}(2),
  301--320.

\bibitem[\protect\citeauthoryear{Zou and Li}{Zou and Li}{2008}]{ZouLi08Onestep}
Zou, H. and R.~Li (2008).
\newblock One-step sparse estimates in nonconcave penalized likelihood models.
\newblock {\em Ann. Statist.\/}~{\em 36\/}(4), 1509--1533.

\end{thebibliography}

\newpage

\section*{Supplementary Materials}

\subsection*{Derivatives of penalty functions}

First two derivatives of commonly used penalty functions are listed in Table \ref{table:pen}.
\begin{table}[h]
\begin{center}
{\small
\begin{tabular}{lrr}
\toprule
Penalty Function & $P_\eta(|\beta|,\rho)$ & $\frac{\partial P_\eta(|\beta|,\rho)}{\partial|\beta|}$  \\
\midrule
Power Family $\eta \in (0,2]$ & $\rho |\beta|^{\eta}$ & $\rho \eta |\beta|^{\eta-1}$  \\
Elastic Net, $\eta \in [1,2]$ & $\rho[(\eta-1) \beta^2 / 2 + (2-\eta) |\beta|]$ & $\rho[(\eta-1)|\beta|+(2-\eta)]$  \\
Log, $\eta>0$ & $\rho \ln (\eta + |\beta|)$ & $\rho (\eta+|\beta|)^{-1}$ \\
Continuous Log & $\rho \ln (\sqrt{\rho}+|\beta|)$ &
    $\rho(\sqrt{\rho} + |\beta|)^{-1}$   \\
SCAD, $\eta>2$ & See (\ref{eqn:SCAD-density}) & $\rho \left\{ 1_{\{|\beta| \le \rho\}} + \frac{(\eta\rho - |\beta|)_+}{(\eta-1)\rho} 1_{\{|\beta| > \rho\}} \right\}$  \\
MC+, $\eta\ge 0$ & See (\ref{eqn:MCP-density}) & $\rho \left(1 - \frac{|\beta|}{\rho \eta} \right)_+$  \\
\bottomrule
\toprule
& $\frac{\partial^2P_\eta(|\beta|,\rho)}{\partial|\beta|^2}$ & $\frac{\partial^2P_\eta(|\beta|,\rho)}{\partial|\beta| \, \partial \rho}$ \\
\midrule
Power Family $\eta \in (0,2]$ & $\rho \eta(\eta-1) |\beta|^{\eta-2}$ & $\eta |\beta|^{\eta-1}$ \\
Elastic Net, $\eta \in [1,2]$ & $\rho(\eta-1)$ & $(\eta-1)|\beta|+(2-\eta)$ \\
Log, $\eta>0$ & $-\rho(\eta+|\beta|)^{-2}$ & $(\eta+|\beta|)^{-1}$ \\
Continuous Log &  $- \rho(\sqrt{\rho} + |\beta|)^{-2}$ & $(\sqrt \rho + \beta)^{-1} - \frac 12 \sqrt \rho (\sqrt \rho + |\beta|)^{-2}$ \\
SCAD, $\eta>2$ & $- (\eta-1)^{-1} 1_{\{|\beta|\in [\rho,\eta\rho)\}}$ & $1_{\{|\beta|<\rho\}} + \eta(\eta-1)^{-1} 1_{\{|\beta|\in [\rho,\eta\rho)\}}$  \\
MC+, $\eta\ge 0$ &  $-\eta^{-1} 1_{\{|\beta|<\rho\eta\}}$ & $1_{\{|\beta|<\rho\eta\}}$  \\
\bottomrule
\end{tabular}
}
\end{center}
\caption{Some commonly used penalty functions and their derivatives.}
\label{table:pen}
\end{table}

\subsection*{Thresholding Formula for Least Squares with Orthogonal Design}

We drop subscript $j$ henceforth to prevent clutter.
\begin{enumerate}
\item For the DP$(\eta)$ penalty, the objective function becomes
\begin{align*}
    \frac{a}{2} (\beta - b)^2 + \rho \ln (\eta + |\beta|).
\end{align*}
Setting derivative to 0, we find that the optimal solution is given by
\begin{align*}
    \hat \beta(\rho)
    &= \begin{cases}
    \frac{(|b|-\eta) + [(|b|+\eta)^2-4\rho/a]^{1/2}}{2} \text{sgn}(b) & \rho \in [0, |ab\eta|]  \\
    0 \text{ or } \frac{(|b|-\eta) + [(|b|+\eta)^2-4\rho/a]^{1/2}}{2} \text{sgn}(b) & \rho \in (|ab\eta|, a(\eta+|b|)^2/4) \\
    0 & \rho \in [a(\eta+|b|)^2/4,\infty)
    \end{cases}.
\end{align*}
The ambiguous case reflects the difficulty with non-convex minimization and has to be resolved by comparing the objective function values at the two points. Moving from one local minimum at $0$ to the other non-zero one results in a jump somewhere in the interval $[|ab\eta|, a(\eta+|b|)^2/4]$. Note that at $\rho = |ab\eta|$, $\beta(\rho) = [(|b|-\eta) + ||b|-\eta|]/2$,  which is zero when $\eta \ge |b|$. Therefore, the path is continuous whenever $\eta \ge |b|$.

\item For the continuous DP penalty, the objective function is
\begin{align*}
    \frac{a}{2} (\beta - b)^2 + \rho \ln (\sqrt{\rho} + |\beta|)
\end{align*}
with the solution
\begin{align*}
    \hat \beta(\rho) &= \frac{(|b|-\sqrt{\rho}) + [(|b|+\sqrt \rho)^2 - 4\rho/a]^{1/2}}{2} \text{sgn}(b)
\end{align*}
for $\rho \le a^2 b^2$. Note at $\rho=a^2b^2$, $\hat \beta(\rho) = [(b-ab)+|b-ab|]/2$, which is 0 when $a \ge 1$. Therefore, when $a \ge 1$, the solution path is continuous
\begin{align*}
    \hat \beta(\rho) &= \begin{cases} \frac{(|b|-\sqrt{\rho}) + [(|b|+\sqrt \rho)^2 - 4\rho/a]^{1/2}}{2} \text{sgn}(b) & \rho \in [0,a^2b^2]  \\
    0 & \rho \in [a^2b^2,\infty)
    \end{cases}.
\end{align*}
When $a<1$, the solution path is given by
\begin{align*}
    \hat \beta(\rho) &= \begin{cases} \frac{(|b|-\sqrt{\rho}) + [(|b|+\sqrt \rho)^2 - 4\rho/a]^{1/2}}{2} \text{sgn}(b) & \rho \in [0,\rho^*]  \\
    0 & \rho \in [\rho^*,\infty)
    \end{cases}
\end{align*}
where $\rho^* > a^2b^2$ indicates the discontinuity point and shall be determined numerically.
\item For power family, the objective function is
\begin{align*}
    \frac{a}{2} (\beta-b)^2 + \rho|\beta|^\eta.
\end{align*}
For $\eta \in (0,1)$, the solution $\hat \beta(\rho)$ is the unique root of the estimating equation
\begin{align*}
    a (\beta-b) + \rho \eta |\beta|^{\eta-1} \text{sgn}(\beta) = 0
\end{align*}
or 0, whichever gives a smaller objective value. For $\eta=1$, it reduces to the well-known soft thresholding operator for lasso $\hat \beta(\rho) = \text{median}\{\pm \rho/a + b, 0\}$. For the convex case $\eta \in (1,2]$, the solution $\hat \beta(\rho)$ is always the (nonzero) root of the estimating equation, i.e., no thresholding; only shrinkage occurs.

\item For elastic net, the objective function is
\begin{align*}
    \frac{a}{2} (\beta-b)^2 + \rho[(\eta-1) \beta^2/2 + (2-\eta)|\beta|]
\end{align*}
with a continuous path solution
\begin{align*}
    \hat \beta(\rho) = \text{median}\left\{ 0, \frac{ab \pm \rho(2-\eta)}{a+\rho(\eta-1)} \right\}.
\end{align*}
Again the lasso soft thresholding is recovered at $\eta=1$ and ridge shrinkage is achieved at $\eta=2$.

\item For SCAD, the minimum of the penalized objective function over $[0,\rho]$ is
\begin{align*}
    \hat \beta_1(\rho) = \text{sgn}(b) \min\{r,u_1\} 1_{r >0},
\end{align*}
where $r = (a|b|-\rho)/a$ and $u_1 = \min \{\rho,|b|\}$, and the minimum over $[\rho,\eta \rho]$ is
\begin{align*}
    \hat \beta_2(\rho) = \text{sgn}(b) \cdot \begin{cases}
    \rho 1_{\{2r>\rho+u_2\}} + u_2 1_{\{2r<\rho+u_2\}} & a(\eta-1)<1 \\
    \rho 1_{\{\eta \rho \ge |b|\}} + u_2 1_{\{\eta \rho < |b|\}} & a(\eta-1)=1  \\
    \rho 1_{\{r \le \rho\}} + r1_{\{r \in [\rho,u_2]\}} + u_2 1_{\{r>u_2\}} & a(\eta-1)>1
    \end{cases},
\end{align*}
where $u_2 = \min\{\eta \rho,|b|\}$ and $r=[ab(\eta-1)-\eta \rho]/[a(\eta-1)-1]$. When $|b|\le \rho$, the solution is $\hat \beta_1(\rho)$. When $|b| \in (\rho,\eta \rho]$, the solution is either $\hat \beta_1(\rho)$ or $\hat \beta_2(\rho)$, whichever gives the smaller penalized objective value. When $|b|>\eta \rho$, the solution is $\hat \beta_1(\rho)$, $\hat \beta_2(\rho)$ or $\hat \beta_3(\rho)=|b|$, whichever gives the smallest penalized objective value.
\item For MC+, the path solution is either
\begin{align*}
    \hat \beta(\rho) = \text{sgn}(b) \cdot \begin{cases}
    b^* 1_{\{2r<b^*\}} & a \eta <1  \\
    b^* 1_{\{\rho<a|b|\}} & a\eta = 1\\
    \min \{r,b^*\} 1_{\{r>0\}} & a\eta > 1
    \end{cases},
\end{align*}
where $b^* = \min \{\rho \eta, |b|\}$ and $r = - (\rho-a|b|)/(a-\eta^{-1})$, or $\hat \beta(\rho) = b$, whichever gives a smaller penalized objective value.
\end{enumerate}

\subsection*{Numerical Comparisons}

Figure \ref{fig:glmnet-timing} displays the run times of lasso penalized GLM by the GLMNet package in {\sc R} \citep{FriedmanHastieTibshirani10GLMCD}, which is the state-of-the-art method for calculating the solution paths of GLM model with enet penalties. It applies coordinate descent algorithm to a sequence of tuning parameters with warm start. The simulation setup is same as in Section \ref{sec:GLM-n200-p10000} and the top 100 predictors are requested from the path algorithm using its default setting. In general GLMNet shows shorter run times than the ODE path algorithm (last column of Figure \ref{fig:n200-p10000-timing}). However a major difference is that GLMNet only computes the solution at a finite number of tuning parameter values (100 by default), while the ODE solver smoothly tracks the whole solution path.

\begin{figure}[!ht]
\begin{center}
$$
\begin{array}{c}
\includegraphics[width=4.5in]{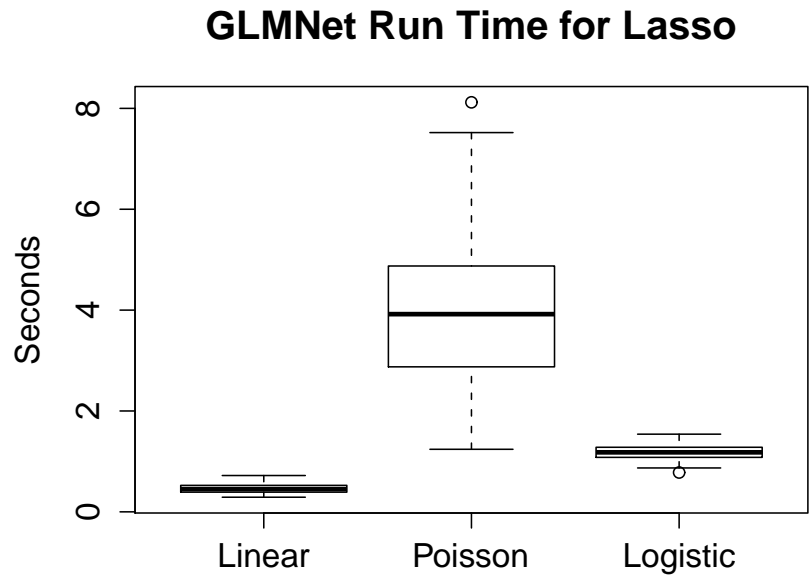}
\end{array}
$$
\vspace{-.5in}
\caption{Run times of GLMNet for the lasso problems from 100 replicates. Problem size is $n=200$ and $p=10,000$. Top 100 predictors are requested.}
\label{fig:glmnet-timing}
\end{center}
\end{figure}

\end{document}